\documentclass[a4paper]{article}
\usepackage[utf8]{inputenc}
\usepackage{a4wide}
\usepackage{amssymb,amsmath,amsthm}
\usepackage{color}

\DeclareMathSymbol{:}{\mathord}{operators}{"3A}

\usepackage{mathtools}

\usepackage{hyperref}
\hypersetup{
    colorlinks,
    linkcolor={red!50!black},
    citecolor={blue!50!black},
    urlcolor={blue!80!black}
}
\usepackage[ocgcolorlinks]{ocgx2}
\usepackage[capitalise,nameinlink]{cleveref}
\usepackage[defaultlines=3,all]{nowidow}

\usepackage{graphicx}
\usepackage{enumitem}

\theoremstyle{plain}
\newtheorem{theorem}{Theorem}
\newtheorem{definition}{Definition}
\newtheorem{corollary}{Corollary}
\newtheorem{claim}{Claim}
\newtheorem{lemma}{Lemma}
\newtheorem{observation}{Observation}
\newtheorem{proposition}{Proposition}

\newenvironment{proofclaim}{
	\noindent \emph{Proof.}
}{%
	\hfill $\diamond$ \\
}

\theoremstyle{remark}

\usepackage{soul}
\theoremstyle{plain}
\theoremstyle{remark}
\usepackage[svgnames,x11names]{xcolor}
\definecolor{mygreen}{HTML}{90ee90}
\usepackage[colorinlistoftodos,backgroundcolor=OrangeRed!80,linecolor=OrangeRed1]{todonotes}

\title{On graphs coverable by $k$ shortest paths}
\author{Maël Dumas\footnote{Univ. Orléans, INSA Centre Val de Loire, LIFO EA 4022, F-45067 Orléans Cedex 2, France.} \and Florent Foucaud\footnote{\noindent Université Clermont Auvergne, CNRS, Clermont Auvergne INP, Mines Saint-Etienne, LIMOS, 63000 Clermont-Ferrand, France.}~\footnote{This author was financed by the ANR project GRALMECO (ANR-21-CE48-0004) and the French government IDEX-ISITE initiative 16-IDEX-0001 (CAP 20-25).} \and Anthony Perez\footnotemark[2] \and Ioan Todinca\footnotemark[2]}
\date{April 2022}

\begin{document}

\maketitle

\newcommand{\SGSR}{\textsc{Strong Geodetic Set with Terminals}}
\newcommand{\SGS}{\textsc{Strong Geodetic Set}}
\newcommand{\IPC}{\textsc{Isometric Path Cover}}
\newcommand{\IPCR}{\textsc{Isometric Path Cover with Terminals}}

\setlength {\marginparwidth }{2.9cm}

\newcommand{\ff}[1]{\textcolor{blue}{#1}}
\newcommand{\md}[1]{\textcolor{red}{#1}}
\newcommand{\io}[1]{\textcolor{magenta}{#1}}
\newcommand{\ap}[1]{\textcolor{orange}{#1}}

\def\eg{{\em e.g.}}
\def\etal{{\em et al.}}
\def\ie{{\em i.e.}~}
\def\cf{{\em cf.}}

\newcommand{\Gak}{\mathcal{G}^e_k}
\newcommand{\Gvk}{\mathcal{G}^v_k}
\newcommand{\B}{\mathcal{B}}
\newcommand{\Z}{\mathcal{Z}}
\newcommand{\PP}{\mathcal{P}}
\newcommand{\CC}{\mathcal{C}}
\newcommand{\CS}{\mathcal{CS}}
\newcommand{\LS}{\mathcal{LS}}
\newcommand{\LL}{\mathcal{L}}

\newcommand{\tw}{\operatorname{tw}}
\newcommand{\pw}{\operatorname{pw}}
\newcommand{\dist}{\operatorname{dist}}
\newcommand{\col}{\operatorname{col}}
\newcommand{\colP}{\operatorname{col}}
\newcommand{\lab}{\operatorname{lab}}
\newcommand{\monochr}{\operatorname{monochr}}
\newcommand{\Colours}{\operatorname{colours}}
\newcommand{\ColoursP}{\operatorname{colours}}
\newcommand{\ColoursS}{\operatorname{colours}}
\newcommand{\MSOL}{\operatorname{MSOL}_2}

\newcommand{\ColoursSigns}{\operatorname{ColoursSignsW}}
\newcommand{\ColoursSignsSet}{\operatorname{ColoursSigns}}
\newcommand{\GreedyColPath}{\operatorname{GreedyColPath}}
\newcommand{\WordColoursSigns}{colours-signs word}
\newcommand{\vertexcol}{vertex colouring}
\newcommand{\signedcolorcode}{signed colour code}
\newcommand{\canon}{\operatorname{CanonPath}}
\newcommand{\Canonization}{\operatorname{Canonize}}
\newcommand{\CanonCode}{\operatorname{CanonCode}}
\newcommand{\CanonPath}{\operatorname{CanonPath}}
\newcommand{\CanonCol}{\operatorname{CanonCol}}

\newcommand{\colarbo}{$k$-labelled branching}
\newcommand{\colarboration}{\operatorname{Colarboration}}

\newcommand{\Pb}[4]{%
\begin{center}
  \begin{tabular}{|l|}%
  \hline
    \begin{minipage}[c]{0.95\textwidth}
      \smallskip%
      \par\noindent%
      #1%
      \par\noindent%
      \textbf{\textsf{Input}}: #2%
      \par\noindent%
      \textbf{\textsf{Question}}: Does there exist %
      #3?%
      \smallskip%
      \par\noindent%
    \end{minipage}
  \\\hline
  \end{tabular}%
\end{center}%
\vspace{0.2cm}
}

\begin{abstract}
We show that if the edges or vertices of an undirected graph $G$ can be covered by $k$ shortest paths, then the pathwidth of $G$ is upper-bounded by a single-exponential function of $k$. As a corollary, we prove that the problem \IPCR{} (which, given a graph $G$ and a set of $k$ pairs of vertices called \emph{terminals}, asks whether $G$ can be covered by $k$ shortest paths, each joining a pair of terminals) is FPT with respect to the number of terminals. The same holds for the similar problem \SGSR{} (which, given a graph $G$ and a set of $k$ terminals, asks whether there exist $\binom{k}{2}$ shortest paths covering $G$, each joining a distinct pair of terminals). Moreover, this implies that the related problems \IPC{} and \SGS{} (defined similarly but where the set of terminals is not part of the input) are in XP with respect to parameter $k$.
\end{abstract}

\section{Introduction}

Path problems such as \textsc{Hamiltonian Path} are among the most fundamental problems in the field of algorithms. \textsc{Hamiltonian Path} can be generalized as the covering problem \textsc{Path Cover}~\cite{ANDREATTA95}, 
where one asks to cover the vertices of an input graph using a prescribed number of paths. The packing variant is \textsc{Disjoint Paths} where, given a set of $k$ pairs of terminal vertices of a graph $G$, one asks whether there are $k$ vertex-disjoint paths in $G$, each joining two paired terminals. \textsc{Disjoint Paths} is a fundamental problem and a precursor to the field of parameterized complexity due to the celebrated fixed-parameter tractable algorithm devised by Robertson and Seymour~\cite{RS95} for the parameter ``number of paths''. We recall that in the field of parameterized algorithms and complexity, one studies \emph{parameterized problems}, whose input $I$ comes together with a parameter $k$. A parameterized problem is said to be FPT (fixed-parameter tractable) if it can be solved in time $f(k)\cdot |I|^{O(1)}$, for some computable function $f$. If the problem can be solved in time $O(|I|^{f(k)})$, it belongs to class $XP$; see, e.g.,~\cite{CFK+15} for more details. 

In this paper, we will not consider arbitrary paths, but \emph{shortest paths}, which are fundamental for many applications. In the problem \textsc{Disjoint Shortest Paths}, given a graph $G$ and $k$ pairs of terminals, one asks whether $G$ contains $k$ vertex-disjoint shortest paths pairwise connecting the $k$ pairs of terminals. This problem was introduced in~\cite{Eilam-Tzoreff98} and recently shown to be polynomial-time solvable for every fixed $k$ by an XP algorithm~\cite{BNRZ21,Lochet21}. 
The problem \IPCR, which we define as follows, is the covering counterpart of \textsc{Disjoint Shortest Paths}. \\

\Pb{\IPCR}{A graph $G$ and $k$ pairs of vertices $(s_1,t_1), \dots, (s_k,t_k)$ called \textit{terminals}.}{a set of $k$ shortest paths, the $i$th path being an $s_i$-$t_i$ shortest path, such that each vertex of $G$ belongs to at least one of the paths}

The name \IPCR{} comes from the related \IPC{} problem where the terminals are not part of the input, which was introduced in~\cite{IPC} in the context of the Cops and Robbers game on graphs (see also~\cite{AF84}). 

\Pb{\IPC}{A graph $G$ and an integer $k$.}{a set of $k$ shortest paths such that each vertex of $G$ belongs to at least one of the paths}

Closely related variants of \IPCR{} and \IPC{} have been studied, in which there are only $k$ terminals, and one asks to find $\binom{k}{2}$ shortest paths joining each pair of terminals. The version without terminals has been called \SGS{} in the literature; we call the version with terminals \SGSR. It was first studied (independently) in~\cite{DIT21,BrazilMaster}, see also~\cite{preprintBresil}. 

\Pb{\SGSR}{A graph $G$ and a set of $k$ vertices of $G$ called \emph{terminals}.}{a set of $\binom{k}{2}$ shortest paths, each path joining a distinct pair of terminals, such that each vertex of $G$ belongs to at least one of the paths}

The variant where the terminals are not given in the input was defined in~\cite{SGS} as follows.

\Pb{\SGS}{A graph $G$ and an integer $k$.}{a set of $k$ terminals and a set of $\binom{k}{2}$ shortest paths, each path joining a distinct pair of terminals, such that each vertex of $G$ belongs to at least one of the paths}

The complexity of these problems has been studied in the literature. It was shown in~\cite{FloManuscrit} that \IPC{} is NP-complete, even for chordal graphs, a class for which the authors of~\cite{FloManuscrit} also provide a constant-factor approximation algorithm. For general graphs, it was shown in~\cite{TG21} that the problem can be approximated in polynomial time within a factor of $O(\log d)$, where $d$ is the diameter of the input graph. It was proven to be polynomial-time solvable on block graphs~\cite{IPC-block}. 
It is shown in~\cite{DIT21} that \SGSR{} is NP-hard. In~\cite{BrazilMaster}, it is shown that this holds even for bipartite graphs of maximum degree~4 or diameter~6, however \SGSR{} is polynomial-time solvable on split graphs, graphs of diameter~2, block graphs, and cactus graphs. We prove here (cf.~\cref{pr:npc}) that \IPCR{} is also NP-complete.

Finally, \SGS{} is known to be NP-hard~\cite{SGS}, even for bipartite graphs, chordal graphs, graphs of diameter $2$ and cobipartite graphs~\cite{BrazilMaster} as well as for subcubic graphs of arbitrary girth~\cite{DIT21}. However it is polynomial-time solvable on outerplanar graphs~\cite{M20}, cactus graphs, block graphs and threshold graphs~\cite{BrazilMaster}. 

All these problems can also be studied in their edge-covering version, where one requires to cover all edges of the input graph by the corresponding shortest paths. For instance, the \textsc{Strong Edge Geodetic Set} problem is studied in~\cite{MKXAT17}. 
The version of \IPC{} where the solution paths are required to be vertex-disjoint has also been studied under the names of \textsc{Isometric Path Partition}~\cite{M20a} and \textsc{Shortest Path Partition}~\cite{FloManuscritPP}.

\paragraph{Our results} 
Our main combinatorial theorems are as follows (see~\cref{sec:nota} for the definition of pathwidth).

\begin{theorem}\label{thm:twGak}
Let $G$ be a graph whose \textbf{edge set} can be covered by at most $k$ shortest paths. Then the pathwidth of $G$ is $O(3^k)$. 
\end{theorem}

\begin{theorem}
\label{thm:twGvk}
Let $G$ be a graph whose \textbf{vertex set} can be covered by at most $k$ shortest paths. Then the pathwidth of $G$ is $O(k\cdot 3^k)$.
\end{theorem}

We actually show that in such a graph $G$, given an arbitrary vertex $a$ and an integer $D$, the number of vertices at distance exactly $D$ from $a$ is upper-bounded by a function of $k$, and the bound does not depend on the size of the input graph. It follows that a very simple linear-time algorithm based on a breadth-first search provides a path decomposition whose width is upper-bounded by twice the aforementioned function of $k$. The complexity of the algorithm computing the path decomposition does not depend on $k$.

Besides the combinatorial bounds, we employ the celebrated theorem of Courcelle~\cite{Courcelle90}, stating that problems expressible in Monadic Second-Order Logic ($\MSOL$) can be solved in linear time for graphs of bounded treewidth (and thus, of bounded pathwidth). More precisely, we reduce the problem \IPCR{} to an optimization problem expressible in $\MSOL$.  The result can also be obtained by dynamic programming but the algorithm would be tedious and not particularly efficient, therefore we prefer the general logic-based framework for further extensions. The running time is linear in $n$, the number of vertices of the graph, but super-exponential in the parameter $k$. Together with~\cref{thm:twGvk}, this implies the following.

\begin{theorem}\label{thm:algoMain}
\IPCR{} and \SGSR{} are FPT when parameterized by the number of terminals. 
\end{theorem}


\begin{corollary}\label{cor:IPC-SGS-XP}
\IPC{} and \SGS{} are in XP when parameterized by the number of paths, respectively terminals.
\end{corollary}

Thanks to the flexibility of Monadic Second-Order Logic and to~\cref{thm:twGak}, our algorithmic results easily extend to the edge-covering versions of our problems, and to variants where we require the paths to be edge-disjoint, or vertex-disjoint like \textsc{Isometric Path Partition}, studied in~\cite{FloManuscritPP,M20a}. The second part of~\cref{thm:algoMain} answers positively a question asked in~\cite{preprintBresil}.

\paragraph{Outline} After some preliminaries in~\Cref{sec:nota}, we prove~\cref{thm:twGak} in~\Cref{sec:edge,sec:vertex}. More specifically,~\Cref{sec:edge} provides the upper bound on the pathwidth of graphs whose edges are coverable by $k$ shortest paths, then the tools are extended to vertex-coverings in the next section. Algorithmic consequences (\cref{thm:algoMain}) are derived in~\Cref{sec:algo}, and we conclude with some open questions.

\section{Preliminaries and notations}\label{sec:nota}

\paragraph{Paths and concatenation operators $\oplus$ and $\odot$}

We refer to~\cite{DiestelB} for usual notations on graphs. In this paper we only consider undirected, unweighted graphs. For simplicity, we assume that our input graph $G=(V,E)$ is connected, though all our combinatorial and algorithmic results extend to non-connected graphs. As usual $N(x)$ denotes the neighbourhood of vertex $x$.

A path $P$ of graph $G=(V,E)$ is a sequence of distinct vertices $(x_1,\dots , x_l)$ such that for each $i, 1 \leq i \leq l-1$, $\{x_i,x_{i+1}\}$ is an edge of the graph. We also say that $P$ is an $x_1$-$x_l$ path. Note that our paths are simple as they do not use twice the same vertex. 
We denote by $V(P)$ the vertices of path $P$, and by $E(P)$ its edges.  Given two vertices $x,y \in V(P)$, we denote by $P[x,y]$ the subpath of $P$ between $x$ and $y$. Let $|P|$ denote the \emph{length} of path $P$, that is, its number of edges. 
The \emph{distance} between two vertices $a$ and $b$ in $G$ is denoted $\dist(a,b)$ and corresponds to the length of a shortest $a$-$b$ path.

Throughout the paper, we will construct paths by concatenation operations. It is convenient to think of our paths as directed: when we speak of an $a$-$b$ path, we think of it as being directed from $a$ to $b$. 

Given two vertex disjoint paths $\nu = (x_1,\dots , x_l)$ and  $\eta = (y_1,\dots , y_t)$ of $G$ such that $\{x_l,y_1\}$ is an edge of $G$, we define the \emph{concatenation operator} $\oplus$ whose result is $\nu \oplus \eta = (x_1, \ldots, x_l, y_1, \dots, y_t)$.  In particular, $|\nu \oplus \eta| = |\nu|+|\eta|+1$. 

We define similarly the \emph{glueing operator} $\odot$ between two paths $\nu = (x_1, \ldots, x_l)$ and $\eta = (x_l, y_1, \ldots, y_t)$ with $V(\nu) \cap V(\mu) = \{x_l\}$ by $\nu \odot \eta = (x_1,\ldots,x_l,y_1,\ldots,y_t)$. 
Note that in this case 
 $|\nu \odot \eta| = |\nu| + |\eta|$.

\paragraph{Path decompositions through breadth-first search}

A path decomposition of $G=(V,E)$ is a sequence $\PP = (X_1,X_2,\dots, X_q)$ of vertex subsets of $G$, called \emph{bags}, such that for every edge $\{x,y\} \in E$ there is at least one bag containing both endpoints, and for every vertex $x \in V$, the bags containing $x$ form a continuous sub-sequence of $\PP$. The width of  $\PP$ is $\max \{|X_i| - 1 \mid 1 \leq i \leq q\}$, and the \emph{pathwidth} $\pw(G)$ of $G$ is the minimum width over all path decompositions of $G$. 

The \emph{treewidth} $\tw(G)$ of graph $G$ is defined similarly (see e.g.~\cite{DiestelB}), using a so-called tree decomposition; for our purpose, we only need to know that for any graph $G$, $\tw(G) \leq \pw(G)$, in particular any path decomposition is also a tree decomposition of the same width. We also need the following folklore lemma on path decompositions. %

\begin{lemma}\label{le:pwBFS}
Let $G = (V,E)$ be a graph, $a$ be a vertex of $G$ and $K$ be an upper bound on the number of vertices of $G$ at distance exactly $D$ from $a$, for any integer $D$.

Then, $\pw(G) \leq 2K-1$. Moreover, a path decomposition of width $2K-1$ can be computed in linear time, by breadth-first search.
\end{lemma}
\begin{proof}
Let $\operatorname{ecc}(a)$ be the eccentricity of vertex $a$ (i.e., $\max_{x \in V} \dist(a,x)$). For any $D$ with $0 \leq D \leq \operatorname{ecc}(a)$ we denote by $\operatorname{Layer}(D)$ the set of vertices at distance exactly $D$ from $a$, i.e. the layers of a breadth-first search on $G$ starting at $a$. Observe that, by taking as bags the unions $\operatorname{Layer}(D) \cup \operatorname{Layer}(D+1)$ of pairs of consecutive layers,  $0 \leq D < \operatorname{ecc}(a)$, and by ordering them according to $D$, we obtain a path decomposition of $G$. Indeed for each edge $\{x,y\}$ both endpoints are in the same layer or in two consecutive layers, thus will appear in the same bag. For each vertex $x$, it appears in at most two bags: if $d = \dist(a,x)$ then $x$ is in bags $\operatorname{Layer}(d-1) \cup \operatorname{Layer}(d)$ and $\operatorname{Layer}(d) \cup \operatorname{Layer}(d+1)$ (or one bag if $d=0$ or $d=\operatorname{ecc}(a)$), and these bags appear consecutively in the decomposition. Since each layer has at most $K$ vertices, the width of this decomposition is at most $2K-1$.
\end{proof}

\section{Edge-covering with $k$ shortest paths}
\label{sec:edge}

We start by proving~\cref{thm:twGak}, upper bounding the pathwidth of graphs $G=(V,E)$ that are edge-coverable by $k$ shortest paths. In this case, there is a simple and elegant encoding of shortest paths leading to a factorial upper bound, given in~\cref{ss:factorial}. This bound is improved to a single-exponential one in~\cref{ss:improving}. We recall that the case of vertex-coverings, which is more technical, will be studied in~\cref{sec:vertex}.

In this section, $G=(V,E)$ denotes a graph whose edge set is coverable by $k$ shortest paths. Let us fix such a set of paths $\mu_1,\dots,\mu_k$, and call them the \emph{base paths} of $G$. All constructions in this section are built on this particular set of base paths (without explicitly recalling it for each lemma, in order to ease the notations).  
We endow each base path $\mu_c,\ 1 \leq c \leq k$ with an arbitrary \emph{direction}. E.g., assuming that the vertices of $G$ are numbered from $1$ to $n$, the direction of path $P$ is from its smallest towards its largest end-vertex. 
A (directed) subpath $\mu_c[x,y]$ of $\mu_c$ is given a positive \emph{sign} $+$ if it follows the direction of $\mu_c$, otherwise it is given a negative sign $-$. 
For each edge $e$ of $G$, let $\Colours(e)$ be the set of all values $c \in \{1\dots, k\}$ such that $e$ is an edge of $\mu_c$.

\paragraph{Good colourings} Let $P$ be an $a$-$b$ path of $G$, from vertex $a$ to vertex $b$. A \emph{colouring} of $P$ is a function $\col : E(P) \rightarrow \{1,\ldots, k\}$ assigning to each edge $e$ of the path one of its colours $\col(e) \in \Colours(e)$. The colouring $\col$ of $P$ is said to be \emph{good} if, for any colour $c$, the set of edges using this colour form a connected subpath $P[x,y]$ of $P$. (Since our paths are simple, this condition entails that $P[x,y] = \mu_c[x,y]$.) A pair $(P,\col)$ formed by a path together with a good colouring is called \emph{well-coloured path}. 

Operator $\odot$ defined in~\cref{sec:nota} naturally extends to coloured paths.
Given a coloured path $(P,\col)$, we simply denote by $(P[x,y],\col)$ its restriction to a subpath $P[x,y]$ of $P$. 
Finally, we define for any path $P$ and any colour $1 \leq c \leq k$ the function $\monochr_c: E(P) \rightarrow \{c\}$. Hence all edges of the coloured path $(P, \monochr_c)$ have colour $c$. 

With these notations, any well-coloured $a$-$b$ path $(P,\col)$ with colours $(c_1,\dots, c_l)$ appearing in this order can be written as: $$(\mu_{c_1}[x_1,x_2], \monochr_{c_1}) \odot (\mu_{c_2}[x_2,x_3], \monochr_{c_2})\odot \ldots \odot (\mu_{c_l}[x_l,x_{l+1}], \monochr_{c_l})$$ for some vertices $a=x_1,x_2,x_3\ldots, x_l, x_{l+1}=b$. In full words, $P[x_i,x_{i+1}]$ are the monochromatic subpaths of $(P,\col)$, coloured $c_i$.

\begin{lemma}[Good colouring lemma]\label{le:goodcol}
For any pair of vertices $a$ and $b$ of $G$, there exists a well-coloured  $a$-$b$ path $(P,\col)$ such that $P$ is a shortest $a$-$b$ path. 

We will simply call $(P,\col)$  a \emph{well-coloured shortest $a$-$b$ path}.
\end{lemma}

\begin{proof}
Among all shortest $a$-$b$ paths, choose one that admits a colouring with a minimum number of monochromatic subpaths. Let $P$ be this path, and $\col$ the corresponding colouring. Assume by contradiction that the colouring is not good. Then there exist three edges $e_1 = \{y_1,z_1\}$, $e_2 = \{y_2,z_2\}$ and $e_3 = \{y_3,z_3\}$, appearing in this order, such that $\col(e_1) = \col(e_3) \neq \col (e_2)$. Assume w.l.o.g. that the vertices appear in the order $y_1, z_1, y_2, z_2, y_3, z_3$ from $a$ to $b$ (note that we may have $z_1 = y_2$ or $z_2=y_3$). 
Let $c = \col(e_1) = \col(e_3)$. Therefore $z_1$ and $y_3$ are on the same base path $\mu_c$. Let $P'$ be the path obtained from $P$ by replacing $P[z_1,y_3]$ by $\mu_c[z_1,y_3]$. First, $P'$ is no longer than $P$, since $\mu_c[z_1,y_3]$ is a shortest possible $z_1$-$y_3$ path of graph $G$ (in particular, $P'$ has no repeated vertices). Second, in $P'$ we can colour all edges of $P'[z_1,y_3]$ with colour $c$, and keep all other colours unchanged. Hence $P'$ has strictly fewer monochromatic subpaths than $P$ --- a contradiction.
\end{proof}

Let $(P,\col)$ be a well-coloured $a$-$b$ path, with colours $(c_1,\ldots,c_l)$ in this order. Recall that each monochromatic subpath $P[x_i,x_{i+1}]$ of $P$, of colour $c_i$, induces a sign on the corresponding base path $\mu_{c_i}$ (positive if $P[x_i,x_{i+1}]$ has the same direction as $\mu_{c_i}$, negative otherwise). Therefore, we can define the \emph{\WordColoursSigns{}} 
$\ColoursSigns(P,\col) = ((c_1,s_1),(c_2,s_2),\dots,(c_l,s_l))$ on the alphabet $\{1,\dots, k\} \times \{+,-\}$, corresponding to the colours and signs of the monochromatic subpaths of $P$, according to the ordering in which these subpaths appear from $a$ to $b$.

\subsection{Warm-up: Factorial  bound}\label{ss:factorial}

Observe that \WordColoursSigns s have at most $k$ letters on an alphabet of size $2k$. Therefore the number of different such words is upper bounded by a function of $k$:

\begin{lemma}\label{le:countwords}
The number of possible \WordColoursSigns s, over all 
well-coloured paths of $G$, is upper bounded by $h(k) = \sum_{l=1}^k 2^l \frac{k!}{(k-l)!}$.
\end{lemma}
\begin{proof}
We claim that the number of \WordColoursSigns s of $l$ letters is upper bounded by $2^l \frac{k!}{(k-l)!}$. 
Observe that the colours form a word of length $l$, on an alphabet of size $k$, without repetition. The number of such words is $\frac{k!}{(k-l)!}$ (e.g., by choosing $l$ letters among the $k$ possible ones, and applying all possible permutations). Since each letter also has a sign in $\{+,-\}$, we multiply this quantity by $2^l$, and the conclusion follows by summing over all possible values of $l$.
\end{proof}

The following crucial lemma implies that, given a start vertex $a$, a distance $D$ and a \WordColoursSigns{} $\omega$, there is at most one vertex $b$ at distance $D$ from $a$, such that the well-coloured shortest $a$-$b$ path respects word $\omega$. This will allow to upper bound the number of vertices at distance $D$ from $a$.

\begin{lemma}[Colours-signs encoding]\label{le:colsign}
Consider two vertices $b$ and $c$ at the same distance from some vertex $a$ of $G$. Let $(P,\col)$ be a well-coloured shortest $a$-$b$ path and $(P',\col')$ be a well-coloured shortest $a$-$c$ path. 
If $\ColoursSigns(P,\col) = \ColoursSigns(P',\col')$, then $b=c$.
\end{lemma}

\begin{proof}
We proceed by induction on the number of letters of the word $\ColoursSigns(P,\col)$. 
Let us denote it by $\omega=((c_1,s_1),(c_2,s_2),\dots,(c_l,s_l))$. 

Let $P[a,x_2]$ (resp. $P'[a,x'_2]$) be the maximal subpath of $P$ (resp. $P'$) of colour $c_1$ starting from $a$. 
Assume w.l.o.g. that $P[a,x_2]$ is at least as long as $P'[a,x'_2]$. Since both are subpaths of $\mu_{c_1}$, starting from $a$ and having the same sign $s_1$ w.r.t. $\mu_{c_1}$, we actually have that $P'[a,x'_2]$ is contained in $P[a,x_2]$, in particular $x'_2$ is between $a$ and $x_2$ in $P$ and in $\mu_{c_1}$.

Observe that, if word $\omega$ has only one letter, $P[a,x_2] = P$ and $P'[a,x'_2] = P'$, thus they are all of the same length. Since they are of the same sign w.r.t. $\mu_{c_1}$, this implies that $x_2 = x'_2=b=c$,  which proves the base case of our induction. 

Assume now that $\omega$ has $l \geq 2$ letters and that the lemma is true for words of length $l-1$.

Consider first the case when $P[a,x_2]$ and $P'[a,x'_2]$ have the same length. Then $x_2 = x'_2$ is also the first vertex of 
the subpaths of colour $c_2$ of both $P$ and $P'$. Then $(P[x'_2,b],\col)$ and $(P'[x'_2,c],\col')$ are well-coloured shortest paths of the same length, and have the same  \WordColoursSigns{} $((c_2,s_2),\dots,(c_l,s_l))$, with $l-1$ letters. Hence the property follows by the induction hypothesis.

We now handle the second and last case, when  $P[a,x_2]$ is strictly longer than $P'[a,x'_2]$. 
Let $x_3$ be the last vertex of the subpath coloured $c_2$ in $(P,\col)$. In particular, $x'_2,x_2$ and $x_3$ are all vertices of $\mu_{c_2}$. 
Let us make an easy but crucial observation: on $\mu_{c_2}$, vertex $x_2$ is between $x'_2$ and $x_3$.
To prove this claim, note that in path $P$, vertices $a$, $x'_2$ and $x_2$ appear in this order (as observed in the beginning of the proof), and by construction $x_2$ appears between $a$ and $x_3$. Therefore $a,x'_2,x_2,x_3$ appear in this order on $P$, which is a shortest path. Hence $dist(x'_2,x_3) = dist(x'_2,x_2)+dist(x_2,x_3)$. Since the three vertices $x'_2,x_2,x_3$ are all on the shortest path $\mu_{c_2}$, they must appear in this order on it. Consequently, $\mu_{c_2}[x'_2,x_3]$ induces the same sign $s_2$ on $\mu_{c_2}$ as $P[x_2,x_3] = \mu_{c_2}[x_2,x_3]$. 
In particular, in path: $$(P[x'_2,b],\col) = (P[x'_2,x_2],\monochr_{c_1}) \odot (P[x_2,x_3], \monochr_{c_2}) \odot (P[x_3,b],\col),$$ we can replace the first subpath $P[x'_2,x_2]$ coloured $c_1$ by $\mu_{c_2}[x'_2,x_2]$, coloured $c_2$, without changing the total length.
We obtain the well-coloured shortest $x'_2$-$b$ path
 $$(\tilde{P}[x'_2,b],\tilde{col}) = (\mu_{c_2}[x'_2,x_3],\monochr_{c_2})\odot (P[x_3,b],\col).$$
Its \WordColoursSigns{}  
is $((c_2,s_2),\ldots,(c_l,s_l))$, the same as for the shortest $x'_2$-$c$ path $(P'[x'_2,c],\col')$. Moreover, the two paths have the same length, $|P| - |P[a,x'_2]|$, hence by the induction hypothesis we have $b=c$, which proves our lemma.
\end{proof}

\begin{corollary}\label{co:distboundE}
For any vertex $a$ of $G$ and any integer $D$, there are at most $h(k) = \sum_{l=1}^k 2^l \frac{k!}{(k-l)!}$ vertices at distance exactly $D$ from $a$.
\end{corollary}

\begin{proof}
For any fixed vertex $a$ and fixed integer $D$, thanks to~\cref{le:colsign} the number of vertices $x$ at distance exactly $D$ from $a$ is upper-bounded by the number of \WordColoursSigns s, which is in turn upper bounded by $h(k)$ by~\cref{le:countwords}.
\end{proof}

\cref{co:distboundE} together with~\cref{le:pwBFS} entail a weaker version of~\cref{thm:twGak}: the pathwidth of graphs edge-coverable by $k$ shortest paths is at most $2h(k)-1$, providing a factorial upper bound. 

\subsection{Single exponential bound}
\label{ss:improving}

In this section, we will show that the number of vertices at distance $D$ from a vertex $a$ can actually be bounded by $O(3^k)$.
In the previous section we showed that any two shortest well-coloured paths of length $D$ starting from $a$ with the same colours-signs words lead to the same vertex. We observe that paths with different colours-signs word may also lead to the same vertex, see~\cref{fig:intuition_colarbo}. In this section, we generalize this idea to a set of shortest paths from $a$ to all vertices at distance $D$ from $a$. 

Before getting into details, let us try to give an informal description of our construction (see~\cref{fig:colarboration}). Let $\PP$ be a family of well-coloured shortest paths starting from vertex $a$, with different endpoints. Initially this family reaches all vertices at distance $D$ from $a$. We transform it into another family with the same properties, preserving the endpoints. 
Let $x_n$ be the furthest vertex from $a$, contained on all paths of $\PP$, such that these paths share the same coloured subpath from $a$ to $x_n$. We focus on the colours $\CC$ appearing on the edges right after $x_n$. Firstly, we apply the following transformation. For each path $(P,\col) \in \PP$, we consider the furthest edge having its colour $c$ in $\CC$ (e.g., in~\cref{fig:colarboration}, for path $P_3$ this edge is coloured red), and we replace the whole subpath of $P$ from $x_n$ to this edge with a subpath of $\mu_c$, of colour $c$. This new family of paths is then split into subfamilies according to the colour and sign of the first edge after $x_n$, and we branch on each subfamily. We eventually prove that this branching process corresponds to a tree (that we call \emph{\colarbo{} tree}), in the sense that the output paths are bijectively mapped on the leaves of the tree. Moreover we will upper bound the number of leaves of the tree, and thus the number of paths, to $O(3^k)$.

We firstly define \colarbo{} trees and show in~\cref{lem:colarbo} that they have $O(3^k)$ leaves. Then we formally describe the above branching process in~\cref{lem:boundak}, deducing the same upper bound on the number of vertices at distance $D$ from $a$. 

\begin{figure}
    \centering
    \includegraphics[scale=1.7]{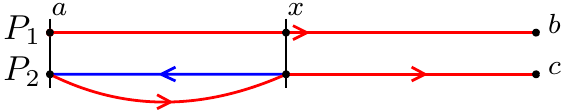}
    \caption{$(P_1,\col_1)$ and $(P_2,\col_2)$ are well-coloured shortest paths of same length. The vertex $x$ is shared by the two paths. One can replace $\mu_{\textcolor{blue}{blue}}[a,x]$ in $P_2$ by $\mu_{\textcolor{red}{red}}[a,x]$, this proves that $b=c$.}
    \label{fig:intuition_colarbo}
\end{figure}


\medskip

Let us define a \emph{rooted labelled tree} as a pair $(T,\lab)$ where $T$ is a rooted tree and $\lab : E(T) \rightarrow \CS$ is an edge labelling function with $\CS = \{1,\dots,k\}\times\{+,-\}$. For a node $n$ of $T$ let $T_n$ be the subtree rooted in $n$.
For a node $n$ of $T$ and its children $n_1, \dots, n_t$ let $\LL(n) = \{\lab(\{n,n_1\}),\dots,\lab(\{n,n_t\})\}$  and $\LS(n) = \{(c,s),(c,\overline{s}) \mid (c,s)\in \LL(n)\}$, where $\overline{s}$ is the opposite of sign $s$.
\begin{definition}
\label{def:colarbo}
A rooted labelled tree $(T,\lab)$ is called a \emph{\colarbo{} tree} 
if the following properties hold:
\begin{enumerate}[label=(\roman*)] 
   \item \label{it:colarbo_path} in every path from the root of $T$ to one of its leaves, the edges with the same label form a connected subpath.
\end{enumerate}

And for every node $n$ of $T$ with children $n_1, \dots, n_t$:
\begin{enumerate}[resume, label=(\roman*)] 
    \item \label{it:colarbo_distinct}
    the labels on the edges $\{n,n_i\}, i\in\{1,\dots,t\}$ are pairwise distinct, 
    \item \label{it:colarbo_subtree} for every $i\in\{1,\dots,t\}$, $T_{n_i}$ does not contain any edge with labels in $\LS(n)\setminus\{\lab(\{n,n_i\})\}$,
   \end{enumerate}
\end{definition}

We now give a combinatorial bound on the number of leaves contained in a \colarbo{} tree. We would like to mention that in a general setting where the labelling function $\lab$ is defined over some alphabet $\Sigma$ and the second condition is replaced by ''$T_{n_i}$ does not contain any edge with labels in $\LL(n)\setminus\{\lab(\{n,n_i\})\}$'', one can achieve an $O(2^{|\Sigma|})$ bound, thus giving $O(4^k)$ in our case (observe that we simply replaced, in the second condition, $\LS$ by $\LL$). However by exploiting the specificity of the alphabet $\CS$ we actually have the following.

\begin{lemma}
\label{lem:colarbo}
A \colarbo{} tree $(T,\lab)$ contains $O(3^k)$ leaves.
\end{lemma}

\begin{proof}

Let $(T,\lab)$ be a \colarbo{} tree, with the labelling function $\lab : E(T) \rightarrow \CS$. 
Let $\B(n)$ be the set of labels that have not been forbidden in $T_n$ by condition~\ref{it:colarbo_path} or condition~\ref{it:colarbo_subtree} of~\cref{def:colarbo} applied to some ancestor of $n$ in the tree $T$. In particular, if $n$ is the root, then $\B(n) = \CS$, and at any node $n$, $\B(n)$ is a superset of all labels of $E(T_n)$.
If $n$ has exactly one child $n'$, by definition $\B(n') \subseteq \B(n)$ so $|\B(n')| \leq |\B(n)|$. 


\begin{claim}\label{claim:budget_colarbo}
    Let $n$ be a node of $T$ with $t\geq 2$ children $n_1, \dots, n_t$. If $n$ is the root then $|\B(n_i)| \leq |\B(n)| - (t - 1)$, else $|\B(n_i)| \leq |\B(n)| - \max\{2,t-1\}$. 
\end{claim}
\begin{proofclaim}    
Since $T$ is a \colarbo{} tree, by condition~\ref{it:colarbo_distinct} the labels of the edges $\{n,n_j\}$, $1 \leq j \leq t$ are pairwise distinct. In particular, the $t-1$ distinct labels of edges $\{n,n_j\}$, with $j \neq i$, appear in $\B(n)$ but not in $\B(n_i)$, by condition~\ref{it:colarbo_subtree}. Therefore $|\B(n_i)| \leq |\B(n)|-(t-1)$.

It remains to prove that if $n$ is not the root and has exactly two children $n_1,n_2$, then $|\B(n_i)| \leq |\B(n)|-2$, for $i \in \{1,2\}$. Let $p$ be the parent node of $n$ in $T$, we denote $(c,s) = \lab(\{p,n\})$ and $(c_i,s_i) = \lab(\{n,n_i\})$ for $i \in \{1,2\}$ (see~\cref{fig:notroot}). 
We claim that if $c_i \neq c$ then both $(c_i,s_i)$ and $(c_i,\overline{s_i})$ are in $\B(n)$. Indeed, in such a case condition~\ref{it:colarbo_path} implies that label $(c_i,s_i)$ does not appear on the path from $n$ to the root of $T$. By condition~\ref{it:colarbo_subtree}, nor does label $(c_i,\overline{s_i})$: this would forbid label $(c_i,s_i)$ in the tree below, in particular in $T_n$. Hence labels $(c_i,s_i)$ and $(c_i,\overline{s_i})$ do not appear on the path from $n$ to the root of $T$. By condition~\ref{it:colarbo_subtree}, for any edge $\{n',n''\}$ of $T$ such that $n'$ is an ancestor of $n$ and $n''$ is not an ancestor of $n$ (see~\cref{fig:notroot}), $c_i$ cannot appear in the label $\lab(\{n',n''\})$ since otherwise $(c_i,s_i)$ and $(c_i,\overline{s_i})$ would be forbidden in the subtree of $T_{n'}$ containing $n$ (which does not contains $n''$). In particular $(c_i,s_i)$ could not be used as label of $\{n,n_i\}$. This means that colour $c_i$ does not appear on edges incident to the path from $n$ to the root of $T$. Altogether both $(c_i,s_i)$ and $ (c_i,\overline{s_i})$ are in $\B(n)$ since they could not be forbidden by conditions~\ref{it:colarbo_path} or~\ref{it:colarbo_subtree}.

\begin{figure}[h]
    \centering
    \includegraphics[scale=1.]{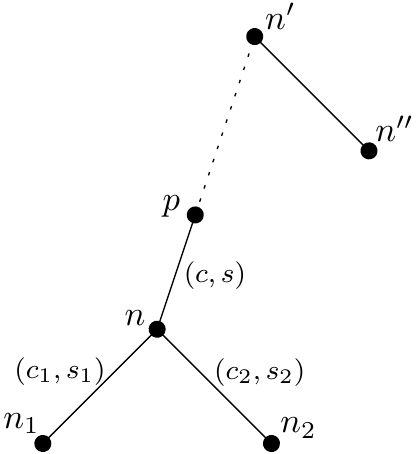}
    \caption{Illustration of the configuration one obtains when $n$ is not the root, has 
    at least two children $n_1$ and $n_2$ and parent $p$. Node $n'$ is an ancestor of $n$ 
    while $n''$ is not.}
    \label{fig:notroot}
\end{figure}

Assume now w.l.o.g. that $i=1$. 
If $c_1 = c$ then we claim that $c_2 \neq c$. Indeed by condition~\ref{it:colarbo_distinct}, if $c_1 = c_2 = c$ then signs $s_1$ and $s_2$ are opposite, so one is the opposite $\overline{s}$ of $s$, contradicting the fact that label $(c,\overline{s})$ is in $\LS(p)$ and has been forbidden in $T_n$ by condition~\ref{it:colarbo_subtree}. Therefore $(c_2,s_2)$ and $(c_2,\overline{s_2})$ are in $\B(n)$ and by condition~\ref{it:colarbo_subtree} are forbidden in $T_{n_1}$ and thus $|\B(n_1)| \leq |\B(n)|-2$. 

If  $c_1 \neq c$, then $(c,s)$ is not in $\B(n_1)$ by condition~\ref{it:colarbo_path}  but is in $\B(n)$  and $(c_1,\overline{s_1})$ is not in $\B(n_1)$ by condition~\ref{it:colarbo_subtree} but is in $\B(n)$, hence $|\B(n_1)| \leq |\B(n)|-2$.
\end{proofclaim}

We now prove that $T$ has $O(3^k)$ leaves. For a node $n$ of the tree, we see $|\B(n)|$ as the \emph{budget} of the tree $T_n$ and we estimate the number of leaves w.r.t. this budget. We remind that the budget of the root is $2k$.
Let $F(b)$ be the maximum number of leaves in any \colarbo{} tree with budget $b$ at its root. Observe that if $b = 1$ then $F(b) = 1$ and for $b\leq b', F(b)\leq F(b')$.
Let $n$ be a node of $T$, then $F(|\B(n)|)$ is the maximal number of leaves of $T_n$. Observe that if $n$ has exactly one child, i.e. it does not branch, then $F(|\B(n)|) = F(|\B(n_1)|)$ since $\B(n_1) \subseteq \B(n)$. Suppose that $n$ has $t\geq 2$ children $n_1,\dots n_t$. If $n$ is the root then, by~\cref{claim:budget_colarbo}: 

$$F(|\B(n)|) \leq F(|\B(n_1)|) + \dots + F(|\B(n_t)|) \leq t \cdot F(|\B(n)| - (t - 1)),$$ else if $n$ is not the root, $$F(|\B(n)|) \leq t \cdot F(|\B(n)| - max\{2,t-1\}).$$

At the root $r$ of the tree, since $\B(r) = \CS$, we have $|\B(r)| = 2k$. We need to prove that $F(2k) = O(3^k)$. 
Let us recall some standard techniques for the analysis on the growth of functions described by linear equations (or linear inequalities). The reader can refer to the book of Fomin and Kratsch~\cite[Chapter 2]{FominK10} for detailed explanations. Consider a function $C$ on positive integers. Assume there are positive integers $\beta_1, \beta_2,\dots,\beta_q$ such that
\begin{equation*}
  C(b) \leq C(b - \beta_1) + C(b - \beta_2) + \dots + C(b-\beta_q),  
\end{equation*}
for all values $b$ (or at least for all $b$ larger than a threshold $b_0$). Then $C(b) = O(\alpha^b)$, where $\alpha$ is the (unique) positive real root of equation

\begin{equation*}
  x^b - x^{b - \beta_1} - x^{b - \beta_2} - \dots - x^{b-\beta_q} = 0.  
\end{equation*}

In~\cite{FominK10}, $(\beta_1, \beta_2,\dots,\beta_q)$ is called a branching vector, and $\alpha$ is the branching factor of this vector. 
Now in our case, for function $F$, we know that for each value $b$ there exists some $t \geq 3$ such that $F(b) \leq t \cdot F(b-(t-1))$, and also (except at the root, if the root has at most two children), $F(b) \leq 2 \cdot F(b-2)$ or $F(b) \leq F(b-1)$. Each of the inequalities leads to a different branching factor: $\alpha_t = \sqrt[t-1]{t}$ for the first equation, and $\alpha_2 = \sqrt{2}$, $\alpha_1= 1$ for the latter. Also according to~\cite{FominK10}, when a function $F$ satisfies one among a family of linear inequalities, the growth of $F$ is $F(b) = O(\alpha^b)$, where $\alpha$ is an upper bound on all branching factors. In our case, observe that the maximum value of $\alpha_1 = 1, \alpha_2 = \sqrt{2}$ and $\alpha_t = \sqrt[t-1]{t}, t\geq 3$ is attained for $t=3$ with $\alpha_3 = \sqrt{3}$. Therefore $F(b) = O(3^{b/2})$. Observe that this also holds if the root has one or two children. Indeed for each child of the root, the number of leaves in the corresponding subtree is $O(3^{b/2})$, and since there are at most two such children, the total number of leaves is  $O(3^{b/2})$.

Since at the root $b = 2k$, we conclude that our \colarbo{} tree has $O(3^k)$ leaves.
\end{proof}

We now prove formally the high-level description given before~\cref{lem:colarbo}. 
In a slight abuse of notation, for a coloured path $(P,\col)$ we define $\ColoursP(P,\col)$ as the set of colours that appear on  $(P,\col)$, \ie{} $\ColoursP(P,\col) = \{\col(e) \mid e \in E(P)\}$. 
For a set of coloured paths $\PP$, $\ColoursS(\PP) = \bigcup_{(P,\col) \in \PP} \ColoursP(P,\col)$.
For a coloured $a$-$b$ path $(P,\col)$ and a vertex $x$ of $P$, we let $P[x:] = P[x,b]$ and for a set of coloured paths $\PP$ that share vertex $x$ we let $\PP[x:]= \bigcup_{(P,\col) \in \PP} (P[x:],\col)$.

\begin{figure}[h]
    \centering
    \includegraphics[scale=1.5]{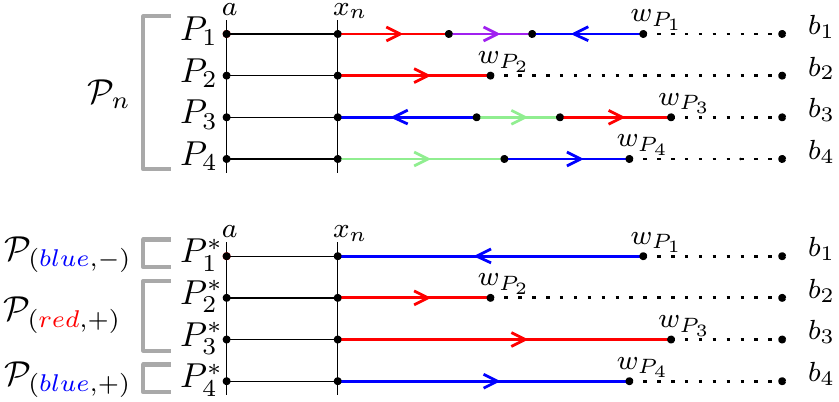}
    \caption{Example of the construction in the proof of~\cref{lem:boundak}
    for a set of paths $\PP_n$. In this example, all paths share the same subpath $(P_1[a,x_n],\col)$, $\CC_n =\{\textcolor{red}{red}, \textcolor{blue}{blue},\textcolor{mygreen}{green}\}$ and there are no edges with colour in $\CC_n$ in the path $(P_i[w_{P_i}:],\col_i)$.    
    The sets $\PP_{(\textcolor{red}{red},-)}, \PP_{(\textcolor{mygreen}{green},+)}, \PP_{(\textcolor{mygreen}{green},-)}$ are empty.}
    \label{fig:colarboration} 
\end{figure}

\begin{lemma}
    For any vertex $a$ of $V(G)$ and any integer $D$, there are  $O(3^k)$ vertices at distance exactly $D$ from $a$.
    \label{lem:boundak}
\end{lemma}

\begin{proof}
Let $\{b_1\dots, b_\ell\}$ be the set of vertices at distance $D$ from $a$ in $G$ and $\PP$ denote a set $\{(P_1,\col_1),\dots, (P_\ell,\col_\ell)\}$ of well-coloured shortest $a$-$b_i$ paths. 
We aim to construct a \colarbo{} tree $(T, \lab)$ such that each node $n$ of $T$ is associated with a set of paths $\PP_n$. Moreover for each vertex $b_i$ there is a leaf $l$ of $T$ such that $\PP_l$ contains a well-coloured shortest $a$-$b_i$ path and for each leaf $l$, $\PP_l$ contains exactly one path.

We will construct $(T,\lab)$ recursively starting from the root $r$ of $T$. Initially $\PP_r = \PP$. For every node $n$ of $T$, $\PP_n$ will be a set of well-coloured shortest paths from $a$ to a subset of vertices of $\{b_1\dots, b_\ell\}$.
Let $n$ be an already computed node of $T$. If $|\PP_n| = 1$ we do nothing:
$n$ is a leaf of our tree and no recursive step is performed at $n$.
If $|\PP_n| > 1$, we will grow the tree from $n$. We construct the sets of paths that will be associated with the children of $n$ as follows:



\begin{enumerate}
    \item Let $x_n$ be the vertex at maximum distance from $a$ such that for every pair of coloured paths $(P,\col),(P',\col') \in \PP_n$, $(P[a,x_n],\col) = (P'[a,x_n],\col')$. This vertex is guaranteed to exist, notice that $x_n$ can be the vertex $a$. 
    \item Let $\CC_n = \{\col(x_n,y_P) \mid (P, \col) \in \PP_n\}$ where $\{x_n,y_P\}$ denotes the first edge of $P[x_n:]$. 
    \item \label{it:constr_path} For each $(P,\col) \in \PP_n$, let $\{z_P,w_P\}$ be the edge that is farthest from $x_n$ on $P[x_n:]$ whose colour is in $\CC_n$, and let $c_P = \col(\{z_P,w_P\})$. Construct the path $$(P^*,\col^*) = (P[a,x_n],\col) \odot
        (\mu_{c_P}[x_n,w_P],\monochr_{c_P}) \odot (P[w_P:],\col)$$
    \item  \label{it:constr_set_path} For each $(c,s)\in \CS_n = (\CC_n\times \{+,-\})$ define the set $\PP_{(c,s)}$ as  $\{(P^*,\col^*) \mid (P, \col) \in \PP_n, (c_P,s_P) = (c,s)\}$, where $s_P$ is the sign of $\mu_{c_P}[x_n,w_P]$.
\end{enumerate}
\vspace{.2cm}

This construction is illustrated~\cref{fig:colarboration}. For each non-empty set $\PP_{(c,s)}, (c,s) \in \CS_n$, we then create a child $n'$ of $n$ associated with $\PP_{(c,s)}$, where the edge $\{n,n'\}$ is labelled $(c,s)$. Then, we recursively apply the process to each created node. 

\begin{claim}
    \label{claim:wcsp}
For any node $n$ of $T$, $\PP_n$ is a set of well-coloured shortest paths. 
\end{claim}

\begin{proofclaim}
    This trivially holds for $\PP_r$ where $r$ is the root of $T$. We show that for a node $n$ of $T$, if $(P,\col)\in \PP_n$ is a well-coloured shortest $a$-$b_i$ path,
    then the path $(P^*, \col^*) = (P[a,x_n],\col) \odot (\mu_{c_P}[x_n,w_P],\monochr_c) \odot (P[w_P:],\col_i)$ constructed in~\cref{it:constr_path} is a well-coloured shortest $a$-$b_i$ path. It is clear from the construction that $P'$ is an $a$-$b_i$ path, moreover it is also a shortest path since $\mu_{c_P}[x_n,w_P]$ is a shortest $x_n$-$w_P$ path. 
    Since $(P,\col)$ is well-coloured, so are $(P[a,x_n],\col)$ and $(P[w_P:],\col)$. 
    Since $\{z_P,w_P\}$ is the last edge of $(P,\col)$ coloured $c_P$, this colour does not appear in $(P[w_P:],\col)$.
    Additionally, if $c_P \in \ColoursP(P[a,x_n])$, it has to be the last colour to appear, otherwise $(P,\col)$ would not be well-coloured. It follows that the path $(P^*,\col^*)$ is well-coloured. Thus, for each $(c,s) \in\CS_n$, the set  $\PP_{(c,s)}$ associated to a child $n'$ of $n$ is a set of well-coloured shortest paths. Hence, by induction, for any node $n$ of $T$, $\PP_n$ is a set of well-coloured-shortest paths. 
\end{proofclaim}

We now show that the construction terminates. In particular, we show that for any non-leaf node $n$ and any of its non-leaf children $n_i$, we have $\dist(a,x_n) < \dist(a,x_{n_i})$.
Since for any node $n$ of $T$, the paths of sets $\PP_n$ are shortest paths of length  $D$, this implies that the construction terminates. 
By construction (see~\cref{it:constr_path}), 
for $(P',\col') \in \PP_{n_i}$, there is a path $(P,\col)\in \PP_n$ such that $(P',\col') = (P^*,\col^*)$, implying that $(P[a,x_n],\col) = (P'[a,x_n],\col')$. 
Moreover, let $\{x_n,y\}$ be the first edge of $(P'[x_n:],\col')$. By definition of  $\PP_{n_i} = \PP_{(c_P,s_P)}$ in~\cref{it:constr_set_path}, all paths in $\PP_{n_i}$ share this edge and hence share the subpath $(P'[a,y],\col')$.
Hence, $\dist(a,x_n) < \dist(a,y) \leq  \dist(a,x_{n_i})$ since the paths of $\PP_n$ and $\PP_{n_i}$ are shortest paths by~\cref{claim:wcsp}.
It is possible to prove that the height of $T$ is at most $k$, but it is not necessary.

The following observation is verified if $n'$ is a child of $n$ by the previous arguments, hence y induction it is also verified for any node of $T_n$. 

\begin{observation}\label{obs:same_subpath}
    Let $n$ be a node of $T$ and $n'\neq n$ a node of $T_n$. Then, we have $\dist(a,x_n) < \dist(a,x_{n'})$, and for $(P,\col)\in \PP_n$ and $(P',\col')\in \PP_{n'}$, $(P[a,x_n],\col) = (P'[a,x_n],\col')$ holds.
\end{observation}

\begin{observation}\label{obs:first_letter}
    Let $\{n,n'\}$ be an edge of $T$ and $(P,col)\in\PP_{n'}$. Then, the first letter of $\ColoursSigns(P[x_n,:],\col)$ is $\lab(\{n,n'\})$. If $n'$ is a non leaf node, there might be more than one letter in  $\ColoursSigns(P[x_n,x_{n'}],\col)$.
\end{observation}

\begin{claim}
\label{claim:colarbo}
 $(T,\lab)$ is a \colarbo{} tree.
\end{claim}
\begin{proofclaim}
Let $n$ be a node of $T$ and $n_1,\dots,n_t$ its children. \cref{it:colarbo_distinct} of~\cref{def:colarbo} is verified, since for each $(c,s) \in \CS_n$, at most one edge $\{n,n_i\}$ is labelled $(c,s)$. 
We now prove~\cref{it:colarbo_subtree}, i.e. $T_{n_i}$ does not contain any edge with labels in $\LS(n)\setminus\{\lab(\{n,n_i\})\}$. Recall that $\LS(n) = \{(c,s),(c,\overline{s}) \mid (c,s)\in \LL(n)\}$, where $\LL(n) = \{\lab(\{n,n_1\}),\dots,\lab(\{n,n_t\})\}$ and $\overline{s}$ is the opposite of $s$.  
Observe that $\LS(n) \subseteq \CS_n$ (there might be some empty set $\PP_{(c,s)}$), and that $\CS_n $ is defined from the set of colours $\CC_n$ that appear on the first edge after $x_n$ on paths of $\PP_n$.

Fix $n_i$ a child of $n$, and let $(c_i,s_i) = \lab(\{n,n_i\})$. The case when $n_i$ is a leaf is trivial, since $T_{n_i}$ has no edges. Assume now that $n_i$ is not a leaf.
Observe that for any $(P,\col) \in \PP_n$ and by choice of $w_P$, the only colour from $\CC_n$ used in $(P^*[x_n:],\col^*)$ is $c_P$. Formally, we have $\ColoursP(P^*[x_n:],\col^*) \subseteq \ColoursP(P[x_n:],\col)\setminus (\CC_n \setminus \{c_P\})$, thus $\ColoursS(\PP_{n_i}[x_n:]) \subseteq \ColoursS(\PP_n[x_n:]) \setminus (\CC_n \setminus \{c_i\})$ for $1 \leqslant i \leqslant t$.
By induction, this implies that for any node $n'$ of $T_{n_i}$, $\ColoursS(\PP_{n'}[x_{n'}:]) \subseteq (\ColoursS(\PP_n[x_n:]) \setminus (\CC_n \setminus \{c_i\})$. By~\cref{obs:first_letter}, the label of an edge $\{n',n''\}$ of $T_{n_i}$ depends on the colours that appear after $x_{n'}$ in paths of $\PP_{n''}$, hence no edges of $T_{n_i}$ has a label in $(\CC_n \setminus \{c_i\}) \times \{+,-\}$.

We now prove that $(c_i,\overline{s_i})$ is not used as label for edges of $T_{n_i}$. Let $\{n',n''\}$ be an edge of $T_{n_i}$ and assume by contradiction that $\lab(\{n',n''\}) = (c_i,\overline{s_i})$. Moreover, let $(P,\col)\in \PP_{n_i}$ and $(P',\col')\in \PP_{n''}$.   
By~\cref{obs:same_subpath}, $(P[x_n,x_{n_i}],\col) = (P'[x_n,x_{n_i}],\col')$. By~\cref{obs:first_letter}, the first letter of $\ColoursSigns (P'[x_n:],\col')$ is $(c_i,s_i)$ and the first letter of $\ColoursSigns (P'[x_{n'}:],\col')$ is $(c_i,\overline{s_i})$, which is a contradiction since the path $(P',\col')$ is well-coloured (\cref{claim:wcsp}). 
Hence, there are no edges of $T_{n_i}$ with labels in $\CS_n \setminus \{(c_i,s_i)\}$, proving~\cref{it:colarbo_subtree}.

It remains to prove~\cref{it:colarbo_path}, i.e., in every path from the root of $T$ to one of its leaves, the edges with the same label form a connected subpath. 
Let $r=p_1,p_2, \dots, p_t =n$ be the path from the root $r$ of $T$ to a non-leaf node $n$. If there exists a child $n'$ of $n$ such that $\lab(\{n,n'\}) = \lab(\{p_{i-1}, p_i\}) = (c,s), i\in\{2,\dots, t\}$, then we have to show that for any $i< j \leq t$, $\lab(\{p_{j-1}, p_j\}) = (c,s)$. 
Suppose that there are such $n'$ and $i$. For $(P,\col)\in \PP_{p_i}$ and $(P',\col')\in \PP_{n'}$, by~\cref{obs:same_subpath} we have $(P[x_{p_{i-1}},x_{p_i}],\col) = (P'[x_{p_{i-1}},x_{p_i}],\col')$. By~\cref{obs:first_letter}, the first edge of these subpaths is coloured $c$, and so is the first edge of $(P'[x_n:],\col')$. 
Since $(P',\col')$ is well\ap{-}coloured, the whole subpath $(P'[x_{p_{i-1}},x_n],\col')$ is entirely coloured with $c$. 
For any $i< j\leq t$, $(P''[x_{p_{j-1}},x_{p_j}],\col'') = (P'[x_{p_{j-1}},x_{p_j}],\col')$, with $(P'',\col'') \in \PP_{p_j}$, and this implies that the first colour of $(P''[x_{p_{j-1}},x_{p_j}],\col'')$ is $c$. Thus $\lab(\{p_{j-1},p_j\}) = (c,s)$ and  
hence~\cref{it:colarbo_path} is verified. It follows that $(T,\col)$ is a \colarbo{} tree, which concludes the proof of~\cref{claim:colarbo}. 
\end{proofclaim}

For a non-leaf node $n$ of $T$, we can observe that for any $a$-$b_i$ path in $\PP_n$, $b_i \in \{b_1\dots, b_\ell\}$, there is a child $n'$ such that $\PP_{n'}$ contains an $a$-$b_i$ path. 
This implies that for any $b_i$ there is a leaf $l$ of $T$ such that $\PP_l$ contains a well-coloured shortest $a$-$b_i$ path. Since $T$ is a \colarbo{} tree, by~\cref{lem:colarbo}, it has $O(3^k)$ leaves and for each leaf $l$ of $T$, $|\PP_l| = 1$ (otherwise our construction would have created a child for node $l$).  It follows that there are $O(3^k)$ vertices at distance $D$ from $a$.
\end{proof}

In order to complete the proof of~\cref{thm:twGak} and show that $\pw(G) = O(3^k)$, we simply apply~\cref{lem:colarbo} with $K = O(3^k)$. 

\section{Vertex-covering with $k$ shortest paths}
\label{sec:vertex}

In this section, $G=(V,E)$ denotes a graph whose \emph{vertices} can be covered by $k$ shortest paths $\mu_1,\dots,\mu_k$. As before we endow each base path $\mu_c$ with a direction, but now colours are assigned to vertices. We can easily adapt the notions of good colourings of the previous section to these vertex-colourings. Again, for any pair of vertices $a$ and $b$, there is a well-coloured shortest path joining them (\cref{le:goodcolV}), which defines a \WordColoursSigns. But we shall see that now (unlike in the simpler case of edge-coverings), we may have two distinct vertices $b$ and $c$ at the same distance $D$ from $a$, and well-coloured shortest $a$-$b$ and $a$-$c$ paths with the same \WordColoursSigns. More efforts will be needed to recover a (slightly larger) upper bound on the number of vertices at distance $D$ from $a$ (\cref{le:distV}).

\paragraph{Good colourings} For each vertex $v$ of $G$, let $\Colours(v)$ denote the set of indices (colours) $c \in \{1\dots, k\}$ such that $v$ is a vertex of $\mu_c$. Let $P$ be an $a$-$b$ path of $G$, from vertex $a$ to vertex $b$. A \emph{colouring} of $P$ is a function $\col: V(P) \rightarrow \{1, \ldots, k\}$ assigning to each vertex $v$ of $P$ one of its colours $\col(v) \in \Colours(v)$. A coloured path is a pair $(P, \col)$. 
The colouring $\col$ of $P$ is said to be \emph{good} if, for any colour $c$, the subgraph induced by the set of \emph{vertices} using this colour $c$ forms a connected subpath $P[x,y]$ of $P$ (which implies that $P[x,y]=\mu_c[x,y]$). 
A coloured path $(P,\col)$ where $\col$ is a good colouring is called \emph{well-coloured}. 

Operators $\oplus$ and $\odot$ naturally extend to (vertex) coloured paths, with the precaution that  $(\nu,\col)\odot (\eta, \col')$ is defined only when their common vertex $x$, the last of $\nu$ and first of $\eta$, satisfies $\col(x) = \col'(x)$.  
Given a coloured path $(P,\col)$, we again denote by $(P[x,y],\col)$ its restriction to a subpath $P[x,y]$ of $P$. 
For each colour $1 \leq c \leq k$, let now $(P, \monochr_c)$ denote the monochromatic colouring of $V(P)$ with colour $c$. 
With these notations, any well-coloured $a$-$b$ path $(P,\col)$ with colours $(c_1,\dots, c_l)$ is of the form 
\begin{equation*}
    (\mu_{c_1}[a_1,b_1], \monochr_{c_1}) \oplus (\mu_{c_2}[a_2,b_2], \monochr_{c_2})\oplus \ldots \oplus (\mu_{c_l}[a_l,b_l], \monochr_{c_l})
\end{equation*}
for some vertices $a=a_1, b_1, a_2, b_2, \ldots, a_l, b_l = b$. 
Like in the previous section, we have:%

\begin{lemma}\label{le:goodcolV}
For any pair of vertices $a$ and $b$ of $G$, there exists a well-coloured shortest $a$-$b$ path. 
\end{lemma}

\begin{proof}
    Among all shortest $a$-$b$ paths, choose one that admits a colouring with a minimum number of monochromatic subpaths. Let $(P,\col)$ be such a coloured path. 
Assume for a contradiction that the colouring $\col$ is not good. Then 
there exist 
three vertices $x,y$ and $z$, appearing in this order in $P$ such that $\col(x) = \col(z) \neq \col (y)$. Therefore $x$ and $z$ are on the same base path $\mu_c$. Let $P'$ be the path obtained from $P$ by replacing $P[x,z]$ by $\mu_c[x,z]$. 
Notice that $P'$ is no longer than $P$, since $\mu_c[x,z]$ is a shortest $x$-$z$ path of graph $G$. 
Moreover, in $P'$ we can colour all vertices of $P'[x,z]$ with colour $c$, and keep all other colours unchanged. Hence $P'$ has strictly fewer monochromatic subpaths than $P$ --- a contradiction.  
\end{proof}

\paragraph{Colours-signs word} Let $(P, \col)$ be a well-coloured $a$-$b$ path; we recall that we see it as being directed from $a$ to $b$. As in~\cref{sec:edge}, each monochromatic subpath $P'$ of $P$, say of colour $c$, induces a sign ($+$ or $-$) depending on its direction w.r.t. $\mu_c$, if $P'$ has at least two vertices. If $P'$ has a unique vertex, we assign to it sign $+$.
Therefore, we can again define the \emph{\WordColoursSigns{}} $\ColoursSigns(P,\col) = ((c_1,s_1),(c_2,s_2),\dots,(c_l,s_l))$ on the alphabet $\{1,\dots, k\} \times \{+,-\}$, corresponding to the colours and signs of the monochromatic subpaths of $P$ according to the ordering in which these subpaths appear from $a$ to $b$. 

In the case of edge-covering, we had the elegant statement of~\cref{le:colsign}, by which, given a vertex $a$, a \WordColoursSigns{} $\omega$ and a distance $D$, there is a unique vertex $b$ (if any exists) at distance $D$ such that the well-coloured shortest $a$-$b$ path corresponds to this word. 

\begin{figure}[h]
    \centering
    \includegraphics[scale=1.6]{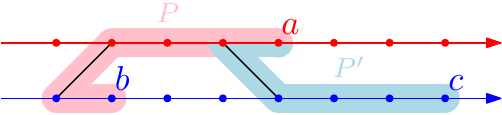}
    \caption{Two well-coloured paths $(P,\col)$ and $(P',\col')$  with same \WordColoursSigns{} $\omega=((red,-),(blue,+))$, same length $5$ and same start vertex $a$ but different end-vertices ($b$ and $c$).}
    \label{fig:difpathsamecol}
\end{figure}

Unfortunately, this does not extend to vertex-covering:~\cref{fig:difpathsamecol} presents two distinct vertices $b$ and $c$ located at the same distance $D$ from vertex $a$, together with a well-coloured shortest $a$-$b$ path $(P,\col)$ and a well-coloured shortest $a$-$c$ path $(P',\col')$. 
These coloured paths starting from $a$ have the same \WordColoursSigns{} and the same length, but this does not imply that their endpoints are equal. However, we can prove a single-exponential bound for this case as well with similar (but more involved) arguments as the ones in~\cref{ss:improving}.


Again, adapting notations to vertex-coloured paths,  for a coloured path $(P,\col)$ we define $\ColoursP(P,\col) = \{\col(x) \mid x \in V(P)\}$, the set of colours that appear on  $(P,\col)$. 
For a set of coloured paths $\PP$, $\ColoursS(\PP) = \bigcup_{(P,\col) \in \PP} \ColoursP(P,\col)$.
For a coloured $a$-$b$ path $(P,\col)$ and a vertex $x$ of $P$, we let $P[x:] = P[x,b]$ and for a set of coloured paths $\PP$ that shares the vertex $x$ we let $\PP[x:]= \bigcup_{(P,\col) \in \PP} (P[x:],\col)$. 

\begin{lemma}\label{le:distV}
    For any vertex $a$ of $G$ and any integer $D$, there are $O(k\cdot3^k)$ vertices at distance exactly $D$ from $a$. 
    \label{lem:boundak_vertex}
\end{lemma}

\begin{proof}
Let $\{b_1\dots, b_\ell\}$ be the set of vertices at distance $D$ from $a$ in $G$ and $\PP$ denote a set $\{(P_1,\col_1),\dots, (P_\ell,\col_\ell)\}$ of well-coloured shortest $a$-$b_i$ paths. 
We aim to construct a \colarbo{} tree $(T, \lab)$ such that each node $n$ of $T$ is associated with a set of paths $\PP_n$ and for each vertex $b_i$ there is a leaf $l$ of $T$ such that $\PP_l$ contains a well-coloured $a$-$b_i$ path. However, unlike in the proof of~\cref{lem:boundak} where any leaf node was associated to exactly one path, the set $\PP_l$ corresponding to any leaf node $l$ can contain up to $2k+1$ paths.

We will construct recursively $(T,\lab)$ starting from the root $r$ of $T$. Initially $\PP_r = \PP$. For every node $n$ of $T$, $\PP_n$ will be a set of well-coloured paths from $a$ to a subset of vertices of $\{b_1\dots, b_\ell\}$. These paths may not be shortest paths, but will be of length at most $D+2d$ where $d$ is the depth of the node $n$ in $T$ (the depth of the root being 0). 
Let $n$ be an already computed node of $T$. 
If there is a path in $\PP_n$ such that all the others paths of $\PP_n$ are coloured subpaths of it, do nothing: $n$ is a leaf of our tree and no recursive step is performed at $n$. If there is no such path, we will grow the tree from $n$. We now construct the sets of paths that will be associated with the children of $n$. 

We first consider the case where $n$ is the root $r$ of $T$. Note that there might exist $(P,\col),(P',\col') \in \PP_n$ such that $\col(a)\neq \col'(a)$. In any case: \\
\begin{enumerate}
    \item Let $\CC_n = \{\col(a) \mid (P, \col) \in \PP_n\}$. 
    \item \label{it:constr_path_vertex_root} For each $(P,\col) \in \PP_n$, let $w_P$ be the last vertex of $P$ such that $c_P = \col(w_P)\in \CC_n$. Construct the path $(P^*,\col^*) = (\mu_{c_P}[a,w_P],\monochr_{c_P}) \odot (P[w_P:],\col)$.
    \item  \label{it:constr_set_path_vertex_root} For each $(c,s)\in \CS_n = (\CC_n\times \{+,-\})$ define the set $\PP_{(c,s)}$ as  $\{(P^*,\col^*) \mid (P, \col) \in \PP_n, (c_P,s_P) = (c,s)\}$, where $s_P$ is the sign of $\mu_{c_P}[a,w_P]$.
\end{enumerate}
\vspace{.2cm}
If $n \neq r$, let $(P_0,\col_0)$ be a longest path of $\PP_n$ and $x_n$ be the last vertex of $P_0$ such that for any path $(P,\col)\in\PP_n$, either $(P_0[a,x_n],\col_0) = (P[a,x_n],\col)$ or $(P,\col)$ is a coloured subpath of $(P_0[a,x_n],\col_0)$. This vertex is guaranteed to exist, and notice that $x_n$ can be the vertex $a$. 
We proceed as follows: \\
    
\begin{enumerate}[resume]
    \item Let $\PP^\circ_n$ be the set of paths of $\PP_n$ that are coloured subpaths of $(P_0[a,x_n],\col_0)$ and $\hat{\PP}_n =\PP_n \setminus \PP^\circ_n$. 
    \item Let $\CC_n = \{\col(z_P) \mid (P, \col) \in \hat{\PP}_n\}$ where $z_P$ denotes the vertex appearing right after $x_n$ on path $P$. 
    \item \label{it:constr_path_vertex} For $(P,\col) \in \hat{\PP}_n$, let $w_P$ be the last vertex of $P$ such that $c_P = \col(w_P)\in \CC_n$. Let $y_{P}$ be the second vertex of $\mu_{c_P}[x_n,w_P]$ if $x_n\in V(\mu_{c_P})$; otherwise $y_P$ is defined as the first vertex of $\mu_{c_P}$ adjacent to $x_n$. We define: 
    $$(P^*,\col^*) = (P[a,x_n],\col) \oplus (\mu_{c_P}[y_P,w_P],\monochr_{c_P}) \odot (P[w_P:],\col)$$
    
    \item  \label{it:constr_set_path_vertex} For each $(c,s)\in \CS_n = (\CC_n\times \{+,-\})$ define the set $\PP_{(c,s)}$ as  $\{(P^*,\col^*) \mid (P, \col) \in \hat{\PP}_n, (c_P,s_P) = (c,s)\}$. Here $s_P$ is
    the sign of $\mu_{c_P}[y_P,w_P]$, except when $y_P = w_P$ and $\col^*(x_n) = \col^*(y_P)$, in which case $s_P$ is the sign of $\mu_{c_P}[x_n,y_P]$.

    \item Add $\PP^\circ_n$ to an arbitrary non-empty set $\PP_{(c,s)}$.
\end{enumerate}
\vspace{.2cm}

This construction is illustrated in~\cref{fig:colarboration_vertex}. For each non-empty set $\PP_{(c,s)}, (c,s) \in \CS_n$ we then create a child $n'$ of $n$ associated with $\PP_{(c,s)}$ where the edge $\{n,n'\}$ labelled $(c,s)$. Then, we recursively apply the process to each created node.

\begin{figure}[h]
    \centering
    \includegraphics[scale=1.5]{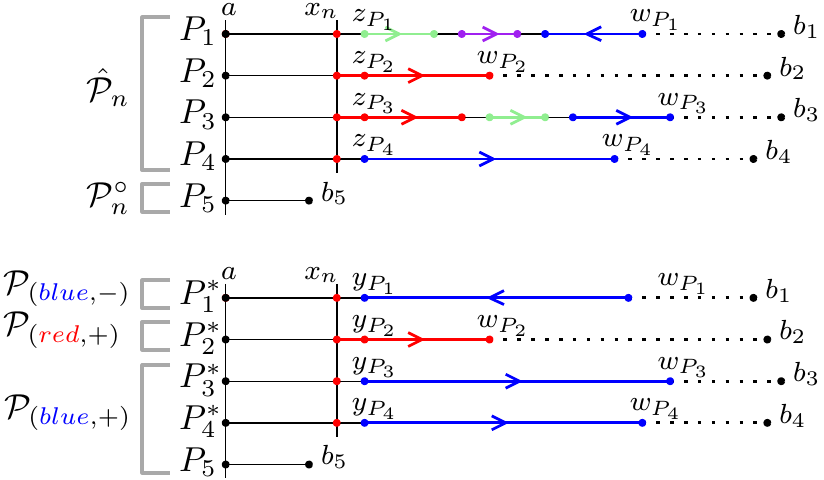}
    \caption{Example of the construction in proof of~\cref{lem:boundak_vertex}
    for a set of paths $\PP_n$. In this example, all paths share the same subpath $(P_1[a,x_n],\col)$ except for $P_5$ that is subpath of the latter. Moreover these paths may have different length. We have $\CC_n =\{\textcolor{red}{red}, \textcolor{blue}{blue},\textcolor{mygreen}{green}\}$ and there are no vertices with colour in $\CC_n$ in the path $(P_i[w_{P_i}:],\col_i)$ except for $w_{P_i}$.   
    The sets $\PP_{(\textcolor{red}{red},-)},\PP_{(\textcolor{mygreen}{green},+)}, \PP_{(\textcolor{mygreen}{green},-)}$ are empty. 
    The vertices $y_{P_1}$ and $y_{P_3} = y_{P_4}$ are the same if $x_n \notin V(\mu_{\textcolor{blue}{blue}})$ and distinct else. 
    The path $P^*_2$ has the same length as $P_1$, the other modified paths might not have the same length than their counterpart, even $P^*_4$ as $z_{P_4}$ might be different from $y_{P_4}$. 
    The path $P_5$ was arbitrarily added to  $\PP_{(\textcolor{blue}{blue},+)}$.}
    \label{fig:colarboration_vertex} 
\end{figure}

\begin{claim}
    \label{claim:wcsp_vertex}
For any node $n$ of depth $d$ of $T$, $\PP_n$ is a set of well-coloured paths of length at most $D+2d$. 
\end{claim}

\begin{proofclaim}
    This trivially holds for $\PP_r$ where $r$ is the root of $T$. We show that for a node $n$ of $T$, if $(P,\col)\in \PP_n$ is a well-coloured $a$-$b$ path, 
    then the path $(P^*, \col^*)$ constructed in~\cref{it:constr_path_vertex_root} or \cref{it:constr_path_vertex} is a well-coloured $a$-$b$ path such that $|P^*|\leq |P|+2$. 
    
    In both cases it is clear from the construction that $P^*$ is an $a$-$b$ path. In the case of~\cref{it:constr_path_vertex_root} where $n=r$, since $(P[w_P:],\col)$ is well-coloured and $\col(w_P) = c_P$ it follows that $(P^*,\col^*)$ is well-coloured. Moreover, $(\mu_{c_P}[a,w_P],\monochr_{c_P})$ is a shortest $a$-$w_P$ path, hence $|P^*| = |P|$.

    In the case of~\cref{it:constr_path_vertex}, since $(P,\col)$ is well-coloured so are $(P[a,x_n],\col)$ and $(P[w_P:],\col)$. 
    The vertex $w_P$ is coloured $c_P$ and if $c_P \in \ColoursP(P[a,x_n])$ it has to be the last colour to appear, otherwise $(P,\col)$ would not be well-coloured. It follows that the path $(P^*,\col^*)$ is well-coloured. 

     We now show that $|P^*|\leq |P| + 2$. This is verified if $x_n\in V(\mu_{c_P})$ since $y_P$ is the second vertex of the shortest $x_n$-$w_P$ path $\mu_{c_P}[x_n,w_P]$ and hence $P^* = P[a,x_n] \odot \mu_{c_P}[x_n,w_P] \odot P[w_P:]$ and $|P^*| \leq |P|$.
    Otherwise, in the case when $x_n \notin V(\mu_{c_P})$, $y_P$ is the first vertex of $\mu_{c_P}$ adjacent to $x_n$. The path $\mu_{c_P}[y_P,w_P]$ is a shortest $y_P$-$w_P$ path, hence $|\mu_{c_P}[y_P,w_P]| \leq |(y_P) \oplus P[x_n,w_P]| \leq 1 + |P[x_n,w_P]|$. This implies the following:
    \begin{align*}
        |P^*| &= |P[a,x_n] \oplus \mu_{c_P}[y_P,w_P] \odot P[w_P:]| \\
        &= |P[a,x_n]| + 1 + |\mu_{c_P}[y_P,w_P]| + |P[w_P:]| \\
        &\leq |P[a,x_n]| + 2 + |P[x_n,w_P]| + |P[w_P:]| = |P| + 2
    \end{align*}

    It remains that for each $(c,s) \in\CS_n$, the set  $\PP_{(c,s)}$ associated to a child $n'$ of $n$ is a set of well-coloured paths of length at most $D'+2$, where $D'$ is the maximum length of paths in $\PP_n$. Hence by induction, for any node $n$ of depth $d$ of $T$, $\PP_n$ is a set of well-coloured paths of length at most $D+2d$. 
\end{proofclaim}

Let $n \neq r$ be a node of $T$. Observe that for $(c,s)\in \CS_n$, all paths of $\PP_{(c,s)}$ share the same successor of $x_n$, the vertex $y_P$ defined in~\cref{it:constr_path_vertex}. 
Thus for every node $n \neq r$ of $T$ with parent $p \neq r$, we can define a vertex $y_n$, the successor of $x_p$ in paths of $\hat{\PP}_n$, see~\cref{fig:def_y_n}. If $n = r$ or $p = r$ we let $y_n = a$.

\begin{figure}[h]
    \centering
    \includegraphics[scale=1.7]{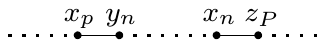}
    \caption{Structure of a path $(P,\col)\in \hat{\PP}_n$ with $p$ the parent of $n$. Note that $y_n = x_n$ is possible.}
    \label{fig:def_y_n}
\end{figure}

Let now $n \neq r$ be a node of $T$ and $n'$ be one of its children. For $(P',\col')\in \hat{\PP}_{n'}$ there is $(P,\col)\in \hat{\PP}_n$ such that 
$(P',\col')$ corresponds to the coloured path $(P^*,\col^*)$ constructed from $(P,\col)$ at node $n$, see~\cref{it:constr_path_vertex}. 
Therefore $(P[a,x_n],\col) = (P'[a,x_n],\col')$. 
Moreover the vertex $y_{n'}$ is the successor of $x_n$ shared among all the paths in $\hat{\PP}_{n'}$. Hence, if $n'$ is a not a leaf, the vertex $x_{n'}$ appears strictly after $x_n$ on paths of  $\hat{\PP}_{n'}$ since  $x_{n'}$  is either $y_{n'}$ or comes after it. 
Therefore, the following observation is verified by induction.

\begin{observation}\label{obs:same_subpath_vertex}
    Let $n \neq r$ be a node of $T$, $n'\neq n$ a node of $T_n$, $(P,\col)\in \hat{\PP}_n$ and $(P',\col')\in \hat{\PP}_{n'}$. Then $(P[a,x_n],\col) = (P'[a,x_n],\col')$ and if $n'$ is a not a leaf, $x_n$ appears strictly before $x_{n'}$ on the path $P'$.
\end{observation}

The following observation is a direct consequence of our construction.
\begin{observation}\label{obs:first_letter_vertex}
    Let $\{n,n'\}$ be an edge of $T$ 
    and $(P,col)\in\hat{\PP}_{n'}$. Then $\lab(\{n,n'\})$ is a letter of $\ColoursSigns(P[a,x_{n'}],\col)$. Moreover, the colour of $\lab(\{n,n'\})$ is the colour of the first letter of $\ColoursSigns(P[y_{n'}:],\col)$ (the signs can be different, see the special case for the definition of $s_P$ in~\cref{it:constr_set_path_vertex}).
\end{observation}

The following claim implies that the construction terminates. 

\begin{claim}\label{claim:height_tree}
    The height of $T$ is at most $k$.
    \end{claim}
\begin{proofclaim}    
Let $n \neq r$ be a node of $T$ and $(P,\col)\in \hat{\PP}_n$. We show that $\ColoursP(P^*[y_P:])\subsetneq \ColoursP(P[y_n:])$. 
From the choice of $w_P$ in~\cref{it:constr_path_vertex}, $\ColoursP(P^*[y_P:])\subseteq \ColoursP(P[y_n:])\setminus (\CC_n\setminus \{c_P\})$. If $|\CC_n| \geq 2$ it follows that $\ColoursP(P^*[y_P:])\subsetneq \ColoursP(P[y_n:])$. 
If $|\CC_n| = 1$, then $c = \col(y_n)\notin \CC_n$, indeed if $c$ was the only colour in $\CC_n$, since paths in $\hat{\PP}_n$ are well-coloured and $y_n\neq y_P$, they all share the subpath $(\mu_c[y_n,y_P],\monochr_c)$, a contradiction with the choice of $x_n$.
Hence $\ColoursP(P^*[y_P:])\subseteq \ColoursP(P[y_n:])\setminus \{c\}$.
It follows that for any child $n'$ of $n$, $\ColoursS(\hat{\PP}_{n'}[y_{n'}:])\subsetneq \ColoursS(\hat{\PP}_n[y_n:])$.
Hence for a node $n\neq r$ of $T$, $|\ColoursS(\hat{\PP}_n[y_n:])|\geq 1$, $|\ColoursS(\hat{\PP}_{n'}[y_{n'}:])|< |\ColoursS(\hat{\PP}_n[y_n:])|$ for a child $n'$ of $n$ and $|\ColoursS(\PP_r)|\leq k$. It follows that the height of $T$ is at most $k$. 
\end{proofclaim}

\begin{claim}
\label{claim:colarbo_vertex}
 $(T,\lab)$ is a \colarbo{} tree.
\end{claim}
\begin{proofclaim}
Let $n$ be a node of $T$ and $n_1,\dots,n_t$ its children.~\cref{it:colarbo_distinct} of~\cref{def:colarbo} is verified since for each $(c,s) \in \CS_n$ at most one edge $\{n,n_i\}$ is labelled $(c,s)$. 

We now prove~\cref{it:colarbo_subtree}, i.e. $T_{n_i}$ does not contain any edge with labels in $\LS(n)\setminus\{\lab(\{n,n_i\})\}$. Recall that $\LS(n) = \{(c,s),(c,\overline{s}) \mid (c,s)\in \LL(n)\}$, where $\LL(n) = \{\lab(\{n,n_1\}),\dots,\lab(\{n,n_t\})\}$ and $\overline{s}$ is the opposite of $s$.  
Observe that $\LS(n) \subseteq \CS_n$ (there might be some empty set $\PP_{(c,s)}$), and that $\CS_n $ is defined by the set of colours $\CC_n$ that appear on the first vertices after $x_n$ on paths of $\hat{\PP}_n$.

Fix $n_i$ a child of $n$ and let $(c_i,s_i) = \lab(\{n,n_i\})$. The case when $n_i$ is a leaf is trivial, since $T_{n_i}$ has no edges, so let us assume that $n_i$ is not a leaf.
For any $(P,\col) \in \hat{\PP}_n$  by choice of $w_P$, 
$\ColoursP(P^*[y_P:],\col^*) \subseteq (\ColoursP(P[y_n:],\col)\setminus (\CC_n \setminus \{c_P\})$ (replace $y_P$ by $a$ if $n$ is the root). Thus $\ColoursS(\hat{\PP}_{n_i}[y_{n_i}:]) \subseteq (\ColoursS(\hat{\PP}_n[y_n:]) \setminus (\CC_n \setminus \{c_i\})$.
By induction, this implies that for any node $n'$ of $T_{n_i}$, $\ColoursS(\hat{\PP}_{n'}[y_{n'}:]) \subseteq (\ColoursS(\hat{\PP}_n[y_n:]) \setminus (\CC_n \setminus \{c_i\})$. 
By~\cref{obs:first_letter_vertex}, the label of an edge $\{n',n''\}$ of $T_{n_i}$ contains the colour of $y_{n''}$ in paths of $\hat{\PP}_{n''}$, hence no edges of $T_{n_i}$ have a label in $(\CC_n \setminus \{c_i\}) \times \{+,-\}$. 
We now prove that $(c_i,\overline{s_i})$ is not used as label for edges of $T_{n_i}$. Let $\{n',n''\}$ be an edge of $T_{n_i}$ and assume that $\lab(\{n',n''\}) = (c_i,\overline{s_i})$. Let $(P,\col)\in \hat{\PP}_{n_i}$ and $(P',\col')\in \hat{\PP}_{n''}$.   
By~\cref{obs:same_subpath_vertex}, $(P[a,x_{n_i}],\col) = (P'[a,x_{n_i}],\col')$, hence by~\cref{obs:first_letter_vertex} $\lab(\{n,n_i\}) = (c_i,s_i)$ is a letter of $\ColoursSigns(P'[a,x_{n_i}],\col')$ and $(c_i,\overline{s_i})$ is a letter of $\ColoursSigns (P',\col')$. This is a contradiction since the path $(P',\col')$ is well-coloured (\cref{claim:wcsp_vertex}). 
Hence there are no edges of $T_{n_i}$ with labels in $\CS_n \setminus \{(c_i,s_i)\}$, proving~\cref{it:colarbo_subtree}.

It remains to prove~\cref{it:colarbo_path}, i.e., in every path from the root of $T$ to one of its leaves, the edges with the same label form a connected subpath. 
Let $r=p_1,p_2, \dots, p_t =n$ be the path from the root $r$ of $T$ to a non-leaf node $n$. If there exists a child $n'$ of $n$ such that $\lab(\{n,n'\}) = \lab(\{p_{i-1}, p_i\}) = (c,s), i\in\{2,\dots, t\}$, then we have to show that for any $i< j \leq t$, $\lab(\{p_{j-1}, p_j\}) = (c,s)$. 
Suppose that there exist such $n'$ and $i$. For $(P,\col)\in \hat{\PP}_{p_i}$ and $(P',\col')\in \hat{\PP}_{n'}$, by~\cref{obs:same_subpath_vertex} $(P[a,x_{p_i}],\col) = (P'[a,x_{p_i}],\col')$. 
By~\cref{obs:first_letter_vertex}, the vertex $y_{p_i}$ (that appears before $x_{p_i}$) is coloured $c$ in both of these subpaths, and so is $y_{n'}$ in $(P',\col')$. 
Since $(P',\col')$ is well-coloured, the subpath $(P'[y_{p_i},y_{n'}],\col')$ is entirely coloured with $c$. 
For any $i< j\leq t$, $(P''[a,x_{p_j}],\col'') = (P'[a,x_{p_j}],\col')$, with $(P'',\col'') \in \hat{\PP}_{p_j}$. This implies that the first colour of $(P''[y_{p_j}:],\col'')$ is $c$. Since there is no edge of $T_{n_i}$ labelled $(c,\overline{s})$, it follows that $\lab(\{p_{j-1},p_j\}) = (c,s)$ and  
hence~\cref{it:colarbo_path} is verified. It remains that $(T,\col)$ is a \colarbo{} tree, which concludes the proof of~\cref{claim:colarbo}.  
\end{proofclaim}

For a non-leaf node $n$ of $T$, we can observe that for any $a$-$b_i$ path in $\PP_n$, $b_i \in \{b_1\dots, b_\ell\}$, there is a child $n'$ such that $\PP_{n'}$ contains a $a$-$b_i$ path (notice that it may be added via the set $\PP^\circ_n$).
This implies that for any $b_i$ there is a leaf $l$ of $T$ such that $\PP_l$ contains a well-coloured $a$-$b_i$ path. For a leaf $l$ of $T$, by construction, $\PP_l$ contains a path $(P,\col)$ including any other path of $\PP_l$. Moreover the length of $P$ is at most $D + 2k$ by Claims~\ref{claim:wcsp_vertex} and~\ref{claim:height_tree}. This implies that $|\PP_l|\leq 2k +1 $, because $\PP_l$ only contains subpaths of $P$, of length between $D$ and $D + 2k$.
Since $T$ is a \colarbo{} tree, by~\cref{lem:colarbo}, it has $O(3^k)$ leaves, hence there are $O(k\cdot 3^k)$ vertices at distance $D$ from $a$.
\end{proof}

In order to complete the proof of~\cref{thm:twGak} and show that $\pw(G) = O(k\cdot 3^k)$, we simply apply~\cref{lem:colarbo} with $K = O(k\cdot 3^k)$. 

\section{Algorithmic consequences} 
\label{sec:algo}

Problem \SGSR\ is known to be NP-complete by~\cite{DIT21}. By a simple reduction, so is \IPCR{}.

\begin{proposition}

\label{pr:npc}
\IPCR{} is NP-Complete.
\end{proposition} 
\begin{proof}
We provide a straightforward reduction from \SGSR. Let $(G = (V,E),k)$ be an instance 
of \SGSR{} with terminals $v_1, \ldots v_k$. 
 We build an instance of \IPCR{} by considering all $\binom{k}{2}$ possible pairs of terminals i.e. $(G = (V,E), \bigcup_{1\leq i < j \leq k} \{(v_i,v_j)\})$. 
By definition, $(G,k)$ is a \textsc{Yes}-instance of \SGSR{}, \ie{}there exists a set of $\binom{k}{2}$ shortest paths covering $V(G)$ if and only if $(G,k')$ 
is a \textsc{Yes}-instance of \IPCR, \ie{} a set of $k'$ shortest $v_i$-$v_j$ paths, $1 \leq i < j \leq k$ covering $V(G)$. 
\end{proof}

\paragraph{Proof of~\cref{thm:algoMain} and~\cref{cor:IPC-SGS-XP}} Recall that, for simplicity, we assume that our input graph is connected, but all results easily extend to disconnected graphs. We first show that problem \IPCR\ is FPT when parameterized by $k$, the number of pairs of terminals. As a first consequence, so is problem \SGSR, a special case of \IPCR{} with $\binom{k}{2}$ pairs of terminals. (Both problems are NP-complete, by~\cite{DIT21} and~\cref{pr:npc}.) 
\cref{cor:IPC-SGS-XP} follows immediately, since for both \IPC{} and \SGS{} it suffices to try all possible sets of terminals and use the FPT algorithms for the versions with terminals.

Let us focus on \IPCR, with parameter $k$, input $G$ and the $k$ pairs of terminals $(s_1,t_1), \dots, (s_k, t_k)$. We can assume that the pathwidth (and hence treewidth) of the input graph is upper bounded by a function of $k$, as stated in~\cref{thm:twGvk}, and that we have in the input a tree decomposition of such width. Indeed recall that~\cref{thm:twGvk} does not only provide a combinatorial bound on the pathwidth of YES-instances, but also a simple, BFS-algorithm for computing the suitable path decompositions (which is also, as stated in~\cref{sec:nota}, a tree decomposition of the same width). If the algorithm fails to find a path decomposition of small width, we can directly conclude that our input graph is a NO-instance. 

Therefore, we can use the classical Monadic Second-Order Logic of graphs (henceforth called $\MSOL$) tools on bounded treewidth graphs. 
$\MSOL$ includes the logical connectives $\vee$, $\land$, $\neg$, 
$\Leftrightarrow$, $\Rightarrow$, variables for 
vertices, edges, sets of vertices, and sets of edges, the quantifiers $\forall$ and $\exists$ that can be applied 
to these variables, and five binary relations: $\operatorname{adj}(u,v)$, where $u$ and $v$ are vertex variables and the interpretation is that $u$ and $v$ are adjacent; $\operatorname{inc}(v,e)$, where $v$ is a vertex variable and $e$ is an edge variable and the interpretation is that $v$ is incident to $e$; $v \in V'$, where $v$ is a vertex variable and $V'$ is a vertex set variable; the similar $e \in E'$ on edge variable $e$ and edge set variable $E'$, and eventually equality of two variables of the same nature. 

By a celebrated theorem of Courcelle~\cite{Courcelle90}, any problem expressible in $\MSOL$ can be solved in time $f(\tw)\cdot n$ time on bounded treewidth graphs, if a tree decomposition of the input graph is also given. Function $f$ depends on the formula (hence, on the problem). Courcelle's theorem extends in several ways to optimization problems, and slightly larger classes of formulae, e.g., allowing to identify a fixed number of terminal vertices, as we shall detail later. Here we will refer to~\cite{ALS91}, one of the (alternative) proofs of Courcelle's theorem, with some extensions. 
As noted in~\cite{ALS91}, $\MSOL$ allows to express properties as $\operatorname{Connected}(V',E')$ where $V'$ is a vertex set variable and $E'$ is an edge set variable and the property is true if and only if $(V',E')$ is a connected subgraph of $G$. Also let $\operatorname{Cover}(V_1,\dots,V_k)$ express the fact that vertex subsets $V_1,\dots, V_k$ cover all vertices of the graph, by simply stating that $\forall x (x \in V_1 \vee x \in V_2 \vee \dots x \in V_k)$. \\ 

This allows us to express \IPCR\ as an optimization $\MSOL$ problem, called an EMS-problem in~\cite{ALS91}. 
Let $\varphi(E_1,E_2,\dots, E_k)$ be the formula on edge sets $E_1, \dots, E_k$ expressing the property that there exist $k$ connected subgraphs $(V_1,E_1), \dots, (V_k,E_k)$ of $G$ such that the sets $V_1, \dots, V_k$ cover all vertices of $G$, and graph $(V_i, E_i)$ contains terminals $s_i, t_i$, for all $1 \leq i \leq k$. More formally:
\begin{eqnarray*}
    \varphi(E_1,E_2,\dots, E_k) & = & \exists\ V_1, V_2, \dots V_k\ [(s_1 \in V_1) \land (t_1 \in V_1) \land \dots \land (s_k \in V_k) \land (t_k \in V_k) \\ 
    & \land & \operatorname{Cover}(V_1,\dots,V_k) \land \operatorname{Connected}(V_1,E_1) \land \dots \land \operatorname{Connected}(V_k,E_k)]
\end{eqnarray*}

Consider now the optimization version of this problem, where the goal is to find edge sets $E_1, E_2, \dots, E_k$ satisfying $\varphi(E_1,E_2,\dots, E_k)$ and minimizing $|E_1| + |E_2| + \dots + |E_k|$. Let $\operatorname{OptCover}$ denote this optimum. By Theorem~5.6 in~\cite{ALS91}, this problem can be solved in linear time on bounded treewidth graphs. More precisely, Arnborg \etal~\cite{ALS91} call such problems "EMS linear extremum problems", in the sense that they correspond to linear optimization functions over the sizes of set variables, when these variables satisfy an $\MSOL$ formula over labelled graphs with a fixed number of labels (here we consider each terminal vertex labelled with a different label). In contemporary terms, the problem is FPT parameterized by treewidth plus the number of terminals, and the running time is linear in $n$.

We now observe that our input is a YES-instance of \IPCR\ if and only if $\operatorname{OptCover} = \dist(s_1,t_1) + \dist(s_2,t_2) + \dots \dist(s_k,t_k)$. Indeed if there exist the required shortest $s_i$-$t_i$ paths $P_1, \dots, P_k$ covering the vertex set of the whole graph, then their edge sets $E_1,\dots, E_k$ provide a solution for our optimization problem whose objective is the sum of the lengths of the paths. Conversely, for any of the $k$ connected subgraphs $(V_i, E_i)$ of $G$ such that $s_i,t_i \in V_i$, we have $|E_i| \geq \dist(s_i,t_i)$. Therefore by simply checking if $\operatorname{OptCover}$ corresponds to the sum of the distances between pairs of terminals, we decide whether the input satisfies \IPCR. Altogether, this problem is FPT parameterized by $k$, which concludes the proof of~\cref{thm:algoMain}.

As mentioned in the introduction, the same techniques extend to variants where the covering paths are required to be edge-disjoint or vertex-disjoint, by simply adding disjointness conditions in the $\MSOL$ formula $\varphi$.

\section{Conclusion}

We have shown that graphs that can be covered by $k$ shortest paths have their pathwidth upper-bounded by a function of $k$. Our bound is exponential, and the first natural open question is whether it can be improved to a polynomial bound. Such an improvement cannot rely on path decompositions based on the layers of an arbitrary BFS, since we have examples (see~\cref{fig:exp_bound}) where the same layer contains $2^k$ vertices. Nevertheless, we leave as an open question whether graphs whose vertices (or edges) can be covered by $k$ shortest paths have treewidth at most a polynomial in $k$.

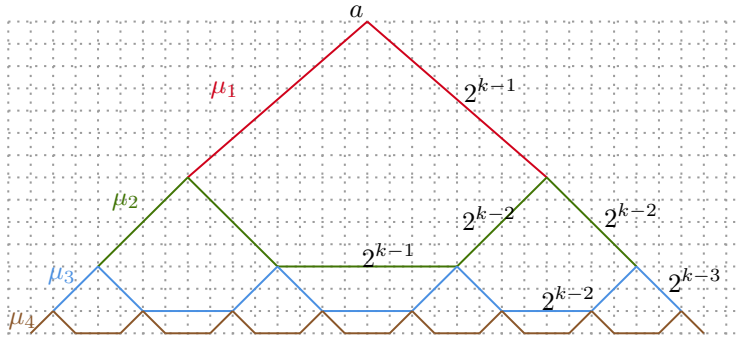
\begin{figure}
    \centering

\tikzset{every picture/.style={line width=0.75pt}} 

\begin{tikzpicture}[x=0.75pt,y=0.75pt,yscale=-0.56,xscale=0.56]

\draw  [draw opacity=0][dash pattern={on 0.84pt off 2.51pt}] (8.5,13) -- (656.5,13) -- (656.5,301) -- (8.5,301) -- cycle ; \draw  [color={rgb, 255:red, 155; green, 155; blue, 155 }  ,draw opacity=1 ][dash pattern={on 0.84pt off 2.51pt}] (8.5,13) -- (8.5,301)(28.5,13) -- (28.5,301)(48.5,13) -- (48.5,301)(68.5,13) -- (68.5,301)(88.5,13) -- (88.5,301)(108.5,13) -- (108.5,301)(128.5,13) -- (128.5,301)(148.5,13) -- (148.5,301)(168.5,13) -- (168.5,301)(188.5,13) -- (188.5,301)(208.5,13) -- (208.5,301)(228.5,13) -- (228.5,301)(248.5,13) -- (248.5,301)(268.5,13) -- (268.5,301)(288.5,13) -- (288.5,301)(308.5,13) -- (308.5,301)(328.5,13) -- (328.5,301)(348.5,13) -- (348.5,301)(368.5,13) -- (368.5,301)(388.5,13) -- (388.5,301)(408.5,13) -- (408.5,301)(428.5,13) -- (428.5,301)(448.5,13) -- (448.5,301)(468.5,13) -- (468.5,301)(488.5,13) -- (488.5,301)(508.5,13) -- (508.5,301)(528.5,13) -- (528.5,301)(548.5,13) -- (548.5,301)(568.5,13) -- (568.5,301)(588.5,13) -- (588.5,301)(608.5,13) -- (608.5,301)(628.5,13) -- (628.5,301)(648.5,13) -- (648.5,301) ; \draw  [color={rgb, 255:red, 155; green, 155; blue, 155 }  ,draw opacity=1 ][dash pattern={on 0.84pt off 2.51pt}] (8.5,13) -- (656.5,13)(8.5,33) -- (656.5,33)(8.5,53) -- (656.5,53)(8.5,73) -- (656.5,73)(8.5,93) -- (656.5,93)(8.5,113) -- (656.5,113)(8.5,133) -- (656.5,133)(8.5,153) -- (656.5,153)(8.5,173) -- (656.5,173)(8.5,193) -- (656.5,193)(8.5,213) -- (656.5,213)(8.5,233) -- (656.5,233)(8.5,253) -- (656.5,253)(8.5,273) -- (656.5,273)(8.5,293) -- (656.5,293) ; \draw  [color={rgb, 255:red, 155; green, 155; blue, 155 }  ,draw opacity=1 ][dash pattern={on 0.84pt off 2.51pt}]  ;
\draw [color={rgb, 255:red, 208; green, 2; blue, 27 }  ,draw opacity=1 ]   (328.5,13) -- (488.5,153) ;
\draw [color={rgb, 255:red, 208; green, 2; blue, 27 }  ,draw opacity=1 ]   (328.5,13) -- (168.5,153) ;
\draw [color={rgb, 255:red, 65; green, 117; blue, 5 }  ,draw opacity=1 ]   (168.5,153) -- (88.5,233) ;
\draw [color={rgb, 255:red, 65; green, 117; blue, 5 }  ,draw opacity=1 ]   (168.5,153) -- (248.5,233) ;
\draw [color={rgb, 255:red, 65; green, 117; blue, 5 }  ,draw opacity=1 ]   (488.5,153) -- (408.5,233) ;
\draw [color={rgb, 255:red, 65; green, 117; blue, 5 }  ,draw opacity=1 ]   (488.5,153) -- (568.5,233) ;
\draw [color={rgb, 255:red, 65; green, 117; blue, 5 }  ,draw opacity=1 ]   (248.5,233) -- (408.5,233) ;
\draw [color={rgb, 255:red, 74; green, 144; blue, 226 }  ,draw opacity=1 ]   (88.5,233) -- (48.5,273) ;
\draw [color={rgb, 255:red, 74; green, 144; blue, 226 }  ,draw opacity=1 ]   (248.5,233) -- (208.5,273) ;
\draw [color={rgb, 255:red, 74; green, 144; blue, 226 }  ,draw opacity=1 ]   (408.5,233) -- (368.5,273) ;
\draw [color={rgb, 255:red, 74; green, 144; blue, 226 }  ,draw opacity=1 ]   (568.5,233) -- (528.5,273) ;
\draw [color={rgb, 255:red, 74; green, 144; blue, 226 }  ,draw opacity=1 ]   (88.5,233) -- (128.5,273) ;
\draw [color={rgb, 255:red, 74; green, 144; blue, 226 }  ,draw opacity=1 ]   (248.5,233) -- (288.5,273) ;
\draw [color={rgb, 255:red, 74; green, 144; blue, 226 }  ,draw opacity=1 ]   (408.5,233) -- (448.5,273) ;
\draw [color={rgb, 255:red, 74; green, 144; blue, 226 }  ,draw opacity=1 ]   (568.5,233) -- (608.5,273) ;
\draw [color={rgb, 255:red, 74; green, 144; blue, 226 }  ,draw opacity=1 ]   (128.5,273) -- (208.5,273) ;
\draw [color={rgb, 255:red, 74; green, 144; blue, 226 }  ,draw opacity=1 ]   (288.5,273) -- (368.5,273) ;
\draw [color={rgb, 255:red, 74; green, 144; blue, 226 }  ,draw opacity=1 ]   (448.5,273) -- (528.5,273) ;
\draw [color={rgb, 255:red, 139; green, 87; blue, 42 }  ,draw opacity=1 ]   (48.5,273) -- (28.5,293) ;
\draw [color={rgb, 255:red, 139; green, 87; blue, 42 }  ,draw opacity=1 ]   (128.5,273) -- (108.5,293) ;
\draw [color={rgb, 255:red, 139; green, 87; blue, 42 }  ,draw opacity=1 ]   (208.5,273) -- (188.5,293) ;
\draw [color={rgb, 255:red, 139; green, 87; blue, 42 }  ,draw opacity=1 ]   (288.5,273) -- (268.5,293) ;
\draw [color={rgb, 255:red, 139; green, 87; blue, 42 }  ,draw opacity=1 ]   (368.5,273) -- (348.5,293) ;
\draw [color={rgb, 255:red, 139; green, 87; blue, 42 }  ,draw opacity=1 ]   (448.5,273) -- (428.5,293) ;
\draw [color={rgb, 255:red, 139; green, 87; blue, 42 }  ,draw opacity=1 ]   (528.5,273) -- (508.5,293) ;
\draw [color={rgb, 255:red, 139; green, 87; blue, 42 }  ,draw opacity=1 ]   (608.5,273) -- (588.5,293) ;
\draw [color={rgb, 255:red, 139; green, 87; blue, 42 }  ,draw opacity=1 ]   (48.5,273) -- (68.5,293) ;
\draw [color={rgb, 255:red, 139; green, 87; blue, 42 }  ,draw opacity=1 ]   (128.5,273) -- (148.5,293) ;
\draw [color={rgb, 255:red, 139; green, 87; blue, 42 }  ,draw opacity=1 ]   (208.5,273) -- (228.5,293) ;
\draw [color={rgb, 255:red, 139; green, 87; blue, 42 }  ,draw opacity=1 ]   (288.5,273) -- (308.5,293) ;
\draw [color={rgb, 255:red, 139; green, 87; blue, 42 }  ,draw opacity=1 ]   (368.5,273) -- (388.5,293) ;
\draw [color={rgb, 255:red, 139; green, 87; blue, 42 }  ,draw opacity=1 ]   (448.5,273) -- (468.5,293) ;
\draw [color={rgb, 255:red, 139; green, 87; blue, 42 }  ,draw opacity=1 ]   (608.5,273) -- (628.5,293) ;
\draw [color={rgb, 255:red, 139; green, 87; blue, 42 }  ,draw opacity=1 ]   (68.5,293) -- (108.5,293) ;
\draw [color={rgb, 255:red, 139; green, 87; blue, 42 }  ,draw opacity=1 ]   (148.5,293) -- (188.5,293) ;
\draw [color={rgb, 255:red, 139; green, 87; blue, 42 }  ,draw opacity=1 ]   (228.5,293) -- (268.5,293) ;
\draw [color={rgb, 255:red, 139; green, 87; blue, 42 }  ,draw opacity=1 ]   (308.5,293) -- (348.5,293) ;
\draw [color={rgb, 255:red, 139; green, 87; blue, 42 }  ,draw opacity=1 ]   (388.5,293) -- (428.5,293) ;
\draw [color={rgb, 255:red, 139; green, 87; blue, 42 }  ,draw opacity=1 ]   (468.5,293) -- (508.5,293) ;
\draw [color={rgb, 255:red, 139; green, 87; blue, 42 }  ,draw opacity=1 ]   (548.5,293) -- (588.5,293) ;
\draw [color={rgb, 255:red, 139; green, 87; blue, 42 }  ,draw opacity=1 ]   (528.5,273) -- (548.5,293) ;

\draw (186,64.4) node [anchor=north west][inner sep=0.75pt]  [color={rgb, 255:red, 208; green, 2; blue, 27 }  ,opacity=1 ]  {$\mu _{1}$};
\draw (98,164.4) node [anchor=north west][inner sep=0.75pt]  [color={rgb, 255:red, 65; green, 117; blue, 5 }  ,opacity=1 ]  {$\mu _{2}$};
\draw (41,231.4) node [anchor=north west][inner sep=0.75pt]  [color={rgb, 255:red, 74; green, 144; blue, 226 }  ,opacity=1 ]  {$\mu _{3}$};
\draw (6,271.4) node [anchor=north west][inner sep=0.75pt]  [color={rgb, 255:red, 139; green, 87; blue, 42 }  ,opacity=1 ]  {$\mu _{4}$};
\draw (310,-4) node [anchor=north west][inner sep=0.75pt]   [align=left] {$\displaystyle a$};
\draw (412,61) node [anchor=north west][inner sep=0.75pt]   [align=left] {$\displaystyle 2^{k-1}$};
\draw (537,173) node [anchor=north west][inner sep=0.75pt]   [align=left] {$\displaystyle 2^{k-2}$};
\draw (410.5,176) node [anchor=north west][inner sep=0.75pt]   [align=left] {$\displaystyle 2^{k-2}$};
\draw (321,209) node [anchor=north west][inner sep=0.75pt]   [align=left] {$\displaystyle 2^{k-1}$};
\draw (594,232) node [anchor=north west][inner sep=0.75pt]   [align=left] {$\displaystyle 2^{k-3}$};
\draw (481,248) node [anchor=north west][inner sep=0.75pt]   [align=left] {$\displaystyle 2^{k-2}$};

\end{tikzpicture}

    \caption{A graph that can be edge-covered with $k$ shortest paths and with $2^k$ vertices at distance $2^k - 1$ from $a$. Therefore, a path decomposition obtained by a BFS from vertex $a$ has width exponential in $k$. Nevertheless, one can easily prove that this graph has pathwidth at most $k$.}
    \label{fig:exp_bound}
\end{figure}

Observe that the approach does not generalize to coverings with few \emph{induced paths}, since grids have arbitrarily large treewidth but are edge-coverable by four induced paths.

On the algorithmic side, we have proved that problems \IPCR\ and \SGSR\ are FPT parameterized by the number of terminals. This directly entails that 
problems \IPC{} and \SGS{} are in XP with respect to the same parameter, by simply enumerating all possible pairs (respectively, sets) of terminals. 
An exciting open question is whether these two problems are FPT. By~\cref{thm:twGvk}, this is equivalent to asking if the problems are FPT when parameterized by the solution size (i.e., number of paths/terminals) + pathwidth. (Indeed, if $k$ is the number of terminals,~\cref{thm:twGvk} ensures that either the pathwidth $\pw$ of the input graph is upper bounded by a function $f(k)$, or we can directly reject the input for being a NO-instance. Therefore, if one of the problems is FPT parameterized by $k + \pw$, we obtain an FPT algorithm parameterized by $k$ as follows. The algorithm checks that $\pw \leq f(k)$ as in~\cref{thm:twGvk}, by a simple breadth-first search from an arbitrary vertex. If the assertion is false, the algorithm rejects. Otherwise it simply remains to apply the algorithm parameterized by $k + \pw$ on parameter $k + f(k)$). Nevertheless, the answer to the question whether these problems are FPT for parameter $k+\pw$ seems non-trivial. At least, while many optimization problems are FPT when parameterized by treewidth/pathwidth, several problems including constraints on distances remain $W[1]$-hard even when parameterized by such structural parameters, plus solution size. We can cite recent hardness results for $d$-\textsc{Scattered Set}~\cite{KLP22}, whose goal is to find a large set of vertices at pairwise distance at least $d$ or, even closer to our problems, \textsc{Geodetic set}~\cite{KeKo20}, where one aims to find a small set of terminals of the input graph 
such that the set of all shortest paths between every pair of terminals covers the graph.

Another natural question is the study of these problems on directed graphs. Note that \textsc{Isometric Path Partition}, the partition version of \IPC, was proved to be W[1]-hard on DAGs for the parameter solution size in~\cite{FloManuscritPP}. Moreover, there are tournaments (orientations of cliques) whose vertices can be covered by a unique shortest path (create a tournament from a directed Hamiltonian path, and orient all other arcs in the reverse direction), in particular the treewidth of the underlying undirected graph is not bounded. A DAG that can be covered by two directed paths but whose underlying undirected graph has large treewidth is the following: consider two disjoint directed paths, and add all arcs from the vertices of the first path to all vertices of the second one. Thus, our results cannot be directly extended to digraphs, not even to DAGs.

\bibliographystyle{abbrv}
\bibliography{main}

\end{document}